\documentclass[11pt,a4paper]{article}

\usepackage[usenames,svgnames,xcdraw,table]{xcolor}
\definecolor{DarkGreen}{rgb}{0.1,0.5,0.1}
\usepackage[backref=page]{hyperref}
\hypersetup{
	colorlinks=true,
	linkcolor=red,
	urlcolor=DarkGreen,
	citecolor=blue
}
\renewcommand*{\backref}[1]{}
\renewcommand*{\backrefalt}[4]{%
    \ifcase #1 (Not cited.)%
    \or        (Cited on page~#2)%
    \else      (Cited on pages~#2)%
    \fi}
\usepackage{amssymb,amsthm,amsmath,mathtools}
\usepackage{thmtools,thm-restate}
\usepackage[round]{natbib}
\usepackage{pifont}
\newcommand{\cmark}{\ding{51}}%
\newcommand{\xmark}{\ding{55}}%
\usepackage[capitalise,noabbrev]{cleveref}
\Crefname{property}{Property}{Properties}
\Crefname{example}{Example}{Examples}
\Crefname{table}{Table}{Tables}

\usepackage{nicefrac}
\usepackage{enumerate}
\usepackage{etoolbox}
\usepackage[margin=1in]{geometry}
\usepackage[T1]{fontenc}
\usepackage[tt=false]{libertine}

\usepackage{url}
\usepackage[utf8]{inputenc}
\usepackage[small]{caption}
\usepackage{graphicx}
\usepackage{booktabs}
\usepackage{mathrsfs,amsfonts,dsfont,authblk}
\usepackage{array,multirow,graphicx,bigdelim}

\usepackage[linesnumbered,lined,boxed,ruled,vlined]{algorithm2e}

\SetKwInOut{Parameters}{Parameters}
\SetKwComment{Comment}{$\triangleright$\ }{}
\SetAlFnt{\small}
\SetAlCapFnt{\small}
\SetAlCapNameFnt{\small}
\SetAlCapHSkip{0pt}
\IncMargin{-\parindent}

\usepackage{tikz}
\usetikzlibrary{fit,calc}
\newcommand*{\tikzmk}[1]{\tikz[remember picture,overlay,] \node (#1) {};\ignorespaces}
\newcommand{\boxit}[1]{\tikz[remember picture,overlay]{\node[yshift=3pt,xshift=4pt,fill=#1,opacity=.25,fit={(A)($(B)+(1.0\linewidth,.8\baselineskip)$)}] {};}\ignorespaces}
\colorlet{mygray}{gray!40}

\makeatletter
\renewcommand{\paragraph}{%
  \@startsection{paragraph}{4}%
  {\z@}{1.0ex \@plus 1ex \@minus .2ex}{-1em}%
  {\normalfont\normalsize\bfseries}%
}
\let\oldnl\nl
\newcommand{\nonl}{\renewcommand{\nl}{\let\nl\oldnl}}
\makeatother

  {\list{}{\leftmargin=#1\rightmargin=#1}\item[]}%
  {\endlist}

\usepackage[labelfont={normalfont,bf},textfont=it]{caption}
\usepackage{subcaption}

\theoremstyle{definition}
\newtheorem{assumption}{Assumption}
\newtheorem{examplex}{Example}
\newenvironment{example}{\pushQED{\qed}\examplex}{\popQED\endexamplex}
\theoremstyle{remark}
\newtheorem{remark}{Remark}

\allowdisplaybreaks

\newcommand{\AlgEQonePO}{\textup{\textsc{Alg-eq\oldstylenums{1}+po}}}

\newcommand{\bin}{\texttt{\textup{bin}}}
\newcommand{\C}{\mathcal{C}}

\newcommand{\e}{\mathbf{e}}
\newcommand{\E}{\mathcal{E}}
\newcommand{\EF}[1]{\ifstrempty{#1}{\textrm{\textup{EF}}}
{\textrm{\textup{EF{#1}}}}}
\newcommand{\EFx}{\textrm{\textup{EFx}}}
\newcommand{\eps}{{\varepsilon}}
\newcommand{\EQ}[1]{\ifstrempty{#1}{\textrm{\textup{EQ}}}
{\textrm{\textup{EQ{#1}}}}}
\newcommand{\EQx}{\textrm{\textup{EQx}}}
\newcommand{\fPO}{\textrm{\textup{fPO}}}

\newcommand{\id}{\texttt{\textup{id}}}

\newcommand{\indeg}{\texttt{\textup{indeg}}}
\newcommand{\I}{\mathcal{I}}
\newcommand{\level}{{\mathrm{level}}}
\newcommand{\Leximin}{\textrm{\textup{Leximin}}}
\newcommand{\M}{{\mathcal{M}}}

\newcommand{\MBB}{\textrm{\textup{MBB}}}
\newcommand{\N}{{\mathbb{N}}}

\newcommand{\NPc}{\textrm{\textup{NP-c}}}
\newcommand{\NPC}{\textrm{\textup{NP-complete}}}
\newcommand{\NPh}{\textrm{\textup{NP-h}}}
\newcommand{\NPH}{\textrm{\textup{NP-hard}}}
\newcommand{\NPhard}{\textup{NP-hard}}
\newcommand{\NW}{\textrm{\textup{NSW}}}
\newcommand{\NSW}{\textrm{\textup{NSW}}}
\renewcommand{\O}{{\mathcal{O}}}
\newcommand{\OPT}{\textrm{\textup{OPT}}}
\newcommand{\outdeg}{\texttt{\textup{outdeg}}}
\newcommand{\p}{\mathbf p}
\renewcommand{\P}{\mathcal{P}}
\newcommand{\Polytime}{\textup{Poly-time}}
\newcommand{\Partition}{\textrm{\textsc{Partition}}}
\newcommand{\poly}{\textrm{\textup{poly}}}
\newcommand{\PO}{\textrm{\textup{PO}}}
\newcommand{\pos}{\texttt{\textup{pos}}}

\newcommand{\Prop}[1]{\ifstrempty{#1}{\textrm{\textup{Prop}}}
{\textrm{\textup{Prop{#1}}}}}

\newcommand{\R}{\mathcal{R}}
\newcommand{\SantaClaus}{\textrm{\textsc{Santa Claus}}}

\newcommand{\ThreePartition}{\textrm{\textsc{3-Partition}}}
\newcommand{\V}{\mathcal{V}}
\newcommand{\W}{\mathcal{W}}
\newcommand{\x}{{\mathrm{x}}}

\newcommand{\y}{{\mathrm{y}}}

\title{Equitable Allocations of Indivisible Goods}

\author[1]{Rupert Freeman}
\author[2]{Sujoy Sikdar}
\author[3]{Rohit Vaish}
\author[4]{Lirong Xia}
\affil[1]{Microsoft Research New York City\\
	{\small\texttt{rupert.freeman@microsoft.com}}}
\affil[2]{Rensselaer Polytechnic Institute\\
	{\small\texttt{sikdas@rpi.edu}}}
\affil[3]{Rensselaer Polytechnic Institute\\ 
	{\small\texttt{vaishr2@rpi.edu}}}
\affil[4]{Rensselaer Polytechnic Institute\\
	{\small\texttt{xial@cs.rpi.edu}}}

\begin{document}

\maketitle

\begin{abstract}
In fair division, \emph{equitability} dictates that each participant receives the same level of utility. In this work, we study equitable allocations of indivisible goods among agents with additive valuations. While prior work has studied (approximate) equitability in isolation, we consider equitability in conjunction with other well-studied notions of fairness and economic efficiency. We show that the Leximin algorithm produces an allocation that satisfies equitability up to any good and Pareto optimality. We also give a novel algorithm that guarantees Pareto optimality and equitability up to one good in pseudopolynomial time. Our experiments on real-world preference data reveal that approximate envy-freeness, approximate equitability, and Pareto optimality can often be achieved simultaneously.
\end{abstract}

\section{Introduction}
\label{sec:Introduction}

We consider fair division problems that require a central planner to divide a set of \emph{goods} among a group of \emph{agents}---each with their own individual preferences over the goods---such that the resulting allocation is fair. How exactly one can certify that an allocation is ``fair'' remains a subject of debate, but the literature suggests two distinct viewpoints. In the first viewpoint, an agent should prefer her bundle of goods to some comparison bundle. The gold standard of fairness here is \emph{envy-freeness}, which says that each agent should prefer her bundle of goods to any other agents' bundle.

In this work, we consider the second viewpoint, in which agents compare their happiness levels, or \emph{utilities}. Here, an allocation is considered fair if the planner is able to make all agents equally well-off. A central fairness notion in this context is \emph{equitability}: An equitable allocation is one where agents derive equal utilities from their assigned shares. Stated differently, an equitable allocation seeks to minimize the disparity between the best-off and the worst-off agents.

Both perspectives have merit, but the practical importance of equitability as a fairness criterion has been highlighted in an experimental study conducted by \citet{HP09envy}. They asked human subjects to deliberate over an assignment of indivisible goods subject to a time limit. It was found that the chosen outcomes were equitable (and Pareto optimal) far more often than they were envy-free. They concluded that equitability is a significant predictor of the perceived fairness of an allocation, often more so than envy-freeness.

Like many other fairness notions, equitability has been traditionally studied for \emph{divisible} goods (also called \emph{cake-cutting}). In this setting, it is known that an equitable allocation always exists \citep{DS61cut,A87splitting}. On the computability side, it is known that no finite procedure can find an (exact) equitable division \citep{PW17lower}, though an $\eps$-equitable division can be computed in a finite number of steps \citep{CP12computability,CP12near}.

For \emph{indivisible} goods, an equitable (\EQ{}) allocation might fail to exist even with two agents and a single good, motivating the need for approximations. To this end, \citet{GMT14near} proposed the notion of \emph{near jealousy-freeness}, under which for any pair of agents, the disparity can be reversed by removing any good from the bundle of the agent with higher utility. We refer to this notion as \emph{equitability up to any good} (\EQx{}) in keeping with the nomenclature for a similar relaxation of envy-freeness \citep{CKM+16unreasonable}. We also study \emph{equitability up to one good} (\EQ1{}), requiring only that inequity can be eliminated by removing \emph{some} good from the higher-utility-agent's bundle. \citet{GMT14near} showed that for additive valuations, an \EQx{} (hence, \EQ1{}) allocation always exists and can be computed in polynomial time. However, they did not study Pareto optimality (\PO{}), a fundamental and often desirable notion of economic efficiency that may still be violated by an (approximately) equitable allocation.

Our work takes a deeper dive into the study of (approximately) equitable allocations of indivisible goods---in conjunction with Pareto optimality as well as other well-studied notions of fairness (envy-freeness and its relaxations)---and considers a host of existence and computational questions. \Cref{tab:Results} provides a comprehensive summary of our results. Some of the highlights are:

\begin{table}[]
\centering
\scriptsize
\begin{tabular}{|cl|c|c|c|c|}
 \hline
 \multicolumn{2}{|c|}{\multirow{2}{*}{\textbf{Guarantee(s)}}} & \multicolumn{2}{c|}{\textbf{Existence Results}} & \multicolumn{2}{c|}{\textbf{Computational Results}}\\
 \cline{3-6}
 & & \emph{general} & \emph{special case} & \emph{general} & \emph{special case} \\
 \hline
  & \EQ{} & \multicolumn{2}{c|}{\xmark{}  even for two agents and one good} & \multicolumn{2}{c|}{strongly \NPc{} even for \id{} (\Cref{prop:IdenticalValuations})} \\
   \cline{3-6}
  & \EQx{} & \multicolumn{2}{c|}{\multirow{2}{*}{\cmark{} (\Cref{prop:EQx_Existence_Computation})}} & \multicolumn{2}{c|}{\multirow{2}{*}{\Polytime{} (\Cref{prop:EQx_Existence_Computation})}} \\
  & \EQ{1} & \multicolumn{2}{c|}{} & \multicolumn{2}{c|}{} \\
  \hline
 \ldelim\{{3}{30pt}[\qquad \ \PO{} +] & \EQ{} & \multicolumn{2}{c|}{\xmark{} even for two agents and one good} &  & \Polytime{} for \bin{} (\Cref{thm:EQ+PO_Polytime_BinaryVals}) \\
    \cline{3-4}
   & \EQx{} & & 
     & \multirow{-2}{*}{strongly \NPh{} (\Cref{rem:EQ1+PO_EQx+PO_strong_NPhardness_GeneralVals})} & \Polytime{} for \bin{} (\Cref{thm:EQ1+EF1+PO_Polytime_BinaryVals}) \\
    \cline{5-5}
  & \EQ{1} & \multirow{-2}{*}{\xmark{} (\Cref{eg:EQ1_PO_Nonexistence_binary_valuations})} & \multirow{-2}{*}{\cmark{} for \pos{} (\Cref{prop:Leximin_EQx+PO_Positive_valuations})} & strongly \NPh{} (\Cref{thm:EQ1+PO_StronglyNPHard_GeneralVals}) &  
  Pseudopoly for \pos{} (\Cref{thm:EQ1+PO_pseudopolynomial}) \\
  \hline
 \ldelim\{{3}{27pt}[\EF{} + \PO{} +] & \EQ{} & \multicolumn{2}{c|}{\multirow{3}{*}{\xmark{} even for two agents and one good}} &  & \Polytime{} for \bin{} (\Cref{rem:EF_EQ_PO_Polytime_BinaryVals}) \\
 \cline{6-6}
  & \EQx{} & \multicolumn{2}{c|}{} &  & \multirow{2}{*}{\NPc{} even for \bin{} (\Cref{rem:EF_EQ1/x_PO_NPHardness_BinaryVals})} \\
  & \EQ{1} & \multicolumn{2}{c|}{} &  & \\
   \cline{6-6}
  \cline{3-4}
 \ldelim\{{3}{33pt}[\EFx{} + \PO{} +] & \EQ{} & \multicolumn{2}{c|}{\multirow{3}{*}{\xmark{} even for \bin{} (\Cref{eg:EQ1_PO_Nonexistence_binary_valuations})}} &  & \Polytime{} for \bin{} (\Cref{rem:EF_EQ_PO_Polytime_BinaryVals}) \\
   \cline{6-6}
  & \EQx{} & \multicolumn{2}{c|}{} &  & \multirow{2}{*}{\Polytime{} for \bin{} (\Cref{thm:EQ1+EF1+PO_Polytime_BinaryVals})} \\
  & \EQ{1} & \multicolumn{2}{c|}{} &  & \\
   \cline{6-6}
 \ldelim\{{3}{32pt}[\EF{1} + \PO{} +] & \EQ{} & \multicolumn{2}{c|}{\multirow{3}{*}{\xmark{} even for \pos{} (\Cref{prop:Nonexistence_EQ1+EF1+PO})}} &  & \Polytime{} for \bin{} (\Cref{rem:EF_EQ_PO_Polytime_BinaryVals}) \\
  \cline{6-6}
  & \EQx{} & \multicolumn{2}{c|}{} &  & \multirow{2}{*}{\Polytime{} for \bin{} (\Cref{thm:EQ1+EF1+PO_Polytime_BinaryVals})} \\
  & \EQ{1} & \multicolumn{2}{c|}{} & \multirow{-9}{*}{strongly \NPh{} (\Cref{cor:EF+EQ+PO_StronglyNPHard_GeneralVals})} & \\
  \hline
\end{tabular}
\caption{Summary of results. For ``Existence Results,'' a \cmark{} denotes guaranteed existence while a \xmark{} indicates that existence might fail for some instance. For ``Computational Results,'' \NPc{}/\NPh{} refers to \NPC{}/\NPH{}. The shorthands \bin{}, \id{}, and \pos{} refer to binary, identical, and strictly positive valuations, respectively.}
\label{tab:Results}
\end{table}

\begin{itemize}
	\item We strengthen the aforementioned result of \citet{GMT14near} to show that an \EQx{} \emph{and} \PO{} allocation always exists for strictly positive valuations~(\Cref{prop:Leximin_EQx+PO_Positive_valuations}). Without the positivity assumption, even an \EQ{1}+\PO{} allocation might fail to exist~(\Cref{eg:EQ1_PO_Nonexistence_binary_valuations}), and finding an \EQ{}+\PO{}/\EQx{}+\PO{}/\EQ{1}+\PO{} allocation becomes strongly \NPH{}~(\Cref{thm:EQ1+PO_StronglyNPHard_GeneralVals,rem:EQ1+PO_EQx+PO_strong_NPhardness_GeneralVals}).

	\item As a step towards making the above existence result constructive, we design a pseudopolynomial-time algorithm that always returns an \EQ{1}+\PO{} allocation for strictly positive valuations~(\Cref{thm:EQ1+PO_pseudopolynomial}).

	\item We construct an instance in which no allocation can be \EQ{1}+\EF{1}+\PO{}~(\Cref{prop:Nonexistence_EQ1+EF1+PO}).\footnote{\EF{1} stands for \emph{envy-freeness up to one good}, which is a (necessary) relaxation of envy-freeness defined for indivisible goods; see \Cref{sec:Preliminaries} for the relevant definitions.} We show that determining whether such an allocation exists is, in general, strongly \NPhard{}~(\Cref{cor:EF+EQ+PO_StronglyNPHard_GeneralVals}), but the special case of binary valuations is efficiently solvable~(\Cref{thm:EQ1+EF1+PO_Polytime_BinaryVals}).
	
	\item We validate our theoretical results via experiments on the data from the popular fair division website \emph{Spliddit}\footnote{\url{http://www.spliddit.org/}} as well as on synthetically generated instances~(\Cref{sec:Experiments}).
\end{itemize}

\paragraph{Related Work}

For \emph{divisible} goods (i.e., cake-cutting), \citet{DS61cut} showed that an equitable division always exists (without providing a bound on the number of cuts). Subsequent work has established the existence of equitable divisions where each agent gets a contiguous piece~\citep{CDP13existence,AD15efficiency,C17existence}.

Equitability has also been studied in combination with other fairness and efficiency notions. It is known that there always exists a cake division that is simultaneously equitable and envy-free~\citep{A87splitting}. However, existence might fail if, in addition, one also requires Pareto optimality~\citep{BJK13n} or contiguous pieces~\citep{BJK06better}. Connections between Pareto optimality and social welfare maximizing equitable divisions have also been studied~\citep{BFL+12maxsum}.

For \emph{indivisible} goods, in addition to the work of~\citet{GMT14near} discussed above, \citet{S19fairly} studies equitable and connected allocations of indivisible goods (i.e., when the goods constitute the vertices of a graph and a feasible allocation assigns every agent a connected subgraph).

\section{Preliminaries}
\label{sec:Preliminaries}

\paragraph{Problem instance}
An \emph{instance} $\langle [n], [m], \V \rangle$ of the fair division problem is defined by a set of $n \in \N$ \emph{agents} $[n] = \{1,2,\dots,n\}$, a set of $m \in \N$ \emph{goods} $[m] = \{1,2,\dots,m\}$, and a \emph{valuation profile} $\V = \{v_1,v_2,\dots,v_n\}$ that specifies the preferences of every agent $i \in [n]$ over each subset of the goods in $[m]$ via a \emph{valuation function} $v_i: 2^{[m]} \rightarrow \N \cup \{0\}$.\footnote{The assumption about integrality of valuations is required only for \Cref{thm:EQ1+PO_pseudopolynomial}. All other positive results (i.e., existence and algorithmic results) hold even in the absence of this assumption. Similarly, all negative results (i.e., non-existence and hardness results) hold even if the valuations are restricted to be integral.} We will assume that the valuation functions are \emph{additive}, i.e., for any agent $i \in [n]$ and any set of goods $S \subseteq [m]$, $v_i(S) \coloneqq \sum_{j \in S} v_i(\{j\})$, where $v_i(\emptyset) = 0$. For a singleton good $j \in [m]$, we will write $v_{i,j}$ instead of $v_i(\{j\})$.

\paragraph{Allocation}
An \emph{allocation} $A \coloneqq (A_1,\dots,A_n)$ is an $n$-partition of the set of goods $[m]$, where $A_i \subseteq [m]$ is the \emph{bundle} allocated to the agent $i$ ($A_i$ is allowed to be an empty set). Given an allocation $A$, the \emph{utility} of agent $i \in [n]$ for the bundle $A_i$ is $v_i(A_i) = \sum_{j \in A_i} v_{i,j}$.

\paragraph{Equitable allocations}
An allocation $A$ is said to be \emph{equitable} (\EQ) if for every pair of agents $i,k \in [n]$, we have $v_i(A_i) = v_k(A_k)$. An allocation $A$ is \emph{equitable up to one good} (\EQ1) if for every pair of agents $i,k \in [n]$ such that $A_k \neq \emptyset$, there exists some good $j \in A_k$ such that $v_i(A_i) \geq v_k(A_k \setminus \{j\})$. An allocation $A$ is \emph{equitable up to any good} (\EQx) if for every pair of agents $i,k \in [n]$ such that $A_k \neq \emptyset$ and for every good $j \in A_k$ such that $v_{k,j} > 0$, we have $v_i(A_i) \geq v_k(A_k \setminus \{j\})$.\footnote{Our results hold analogously for the following variant of \EQx{} due to \citet{GMT14near}: For every pair of agents $i,k \in [n]$ such that $A_k \neq \emptyset$, $v_i(A_i) \geq v_k(A_k \setminus \{j\})$ for every good $j \in A_k$.}

\paragraph{Envy-free allocations}
An allocation $A$ is \emph{envy-free} (\EF{}) if for every pair of agents $i,k \in [n]$, we have $v_i(A_i) \geq v_i(A_k)$. An allocation $A$ is \emph{envy-free up to one good} (\EF{1}) if for every pair of agents $i,k \in [n]$ such that $A_k \neq \emptyset$, there exists some good $j \in A_k$ such that $v_i(A_i) \geq v_i(A_k \setminus \{j\})$. An allocation $A$ is \emph{envy-free up to any good} (\EFx) if for every pair of agents $i,k \in [n]$ such that $A_k \neq \emptyset$ and for every good $j \in A_k$ such that $v_{i,j} > 0$, we have $v_i(A_i) \geq v_i(A_k \setminus \{j\})$. 
The notions of \EF{}, \EF{1}, and \EFx{} are due to \citet{F67resource}, \citet{B11combinatorial},\footnote{
\citet{LMM+04approximately} previously defined a slightly weaker notion than \EF{1}, but their algorithm can, in fact, compute an \EF{1} allocation.} and \citet{CKM+16unreasonable}, respectively.

\paragraph{Pareto optimality}
An allocation $A$ is Pareto dominated by another allocation $B$ if $v_k(B_k) \geq v_k(A_k)$ for every agent $k \in [n]$ with at least one of the inequalities being strict. A \emph{Pareto optimal} (\PO{}) allocation is one that is not Pareto dominated by any other allocation. 

\paragraph{Nash social welfare}
Given an instance $\langle [n], [m], \V \rangle$, the \emph{Nash social welfare} of an allocation $A$ is defined as $\NW(A) \coloneqq \left( \prod_{i \in [n]} v_i(A_i) \right)^{1/n}$. An allocation $A^*$ is called  \emph{Nash optimal} or \emph{MNW} (Maximum Nash Welfare) if it maximizes the Nash social welfare among all allocations.\footnote{\citet{CKM+16unreasonable} define a Nash optimal allocation as one that provides positive utility to the largest set of agents, and subject to that, maximizes the geometric mean of valuations. Our results hold even under this extended definition.}

\paragraph{\Leximin -optimal allocations}
A \Leximin -optimal allocation \citep{DS61cut} is one that maximizes the minimum utility that any agent achieves, subject to which the second-minimum utility is maximized, and so on. The utilities induced by a \Leximin -optimal allocation are unique, although there may exist more than one such allocation.

\section{Results}
\label{sec:Results}

This section presents our theoretical results, summarized in Table~\ref{tab:Results}. We first consider equitability and its relaxations, then consider them in conjunction with Pareto optimality, before finally adding envy-freeness (and its relaxations) to the mix.

\subsection{Existence and Computation of \EQ{}, \EQ{1}, \EQx{}}
We will start by observing that envy-freeness and equitability  (and their corresponding relaxations) become equivalent when the valuations are \emph{identical} (i.e., when, for every good $j \in [m]$, $v_{i,j} = v_{k,j}$ for all $i,k \in [n]$).

\begin{restatable}{prop}{IdenticalValuations}
 \label{prop:IdenticalValuations}
For identical valuations, an allocation is \EF{}/\EF{1}/\EFx{} if and only if it is \EQ{}/\EQ{1}/\EQx{}.
\end{restatable}

It is known that determining whether a given instance has an envy-free (\EF{}) allocation is \NPC{} even for identical valuations (via a straightforward reduction from \Partition{}) \citep{LMM+04approximately}.\footnote{In fact, the problem is strongly \NPC{} due to a similar reduction from \ThreePartition{} \citep{GJ79computers}.} \Cref{prop:IdenticalValuations} implies that the same holds for equitable (\EQ{}) allocations. By contrast, an \EQx{} (and therefore \EQ{1}) allocation always exists and can be efficiently computed (\Cref{prop:EQx_Existence_Computation}) even for non-identical valuations. This result is due to \citet{GMT14near}, who showed the existence of \EQx{} allocations under the more general setting of matroids.

\begin{restatable}[\citealp{GMT14near}]{prop}{EQxExistenceComputation}
 \label{prop:EQx_Existence_Computation}
 An \EQx{} allocation always exists and can be computed in polynomial time.
\end{restatable}

Briefly, \citet{GMT14near} prove \Cref{prop:EQx_Existence_Computation} using a greedy algorithm. In each round, the algorithm assigns a least-happy agent its favorite good from among the remaining goods. Thus, at any stage, the most recent good assigned to an agent is also its least-favorite good in its own bundle. Since each new good is assigned to an agent with the least utility, an allocation that is \EQx{} prior to the assignment continues to be so after it (up to the removal of the most recently assigned good). The claim now follows by induction over the rounds.

\Cref{prop:EQx_Existence_Computation} presents an interesting contrast between the notions of \EQx{} and \EFx{}: An \EQx{} allocation is guaranteed to exist and can be efficiently computed, whereas for \EFx{}, even the question of guaranteed existence is an open problem.

\subsection{Equitability and Pareto Optimality}

We now turn our attention to computing an allocation that is both equitable up to one good and Pareto optimal (we use the shorthand \EQ{1}+\PO{} for such allocations). Unfortunately, such allocations might fail to exist when the valuations are allowed to be zero-valued (\Cref{eg:EQ1_PO_Nonexistence_binary_valuations}). This provides an interesting contrast with the analogous relaxation of envy-freeness; it is known that an allocation satisfying \EF{1} and \PO{} always exists \citep{CKM+16unreasonable,BKV18Finding}.

\begin{example}[\textbf{Non-existence of \EQ{1}+\PO{}}]
Consider an instance with three agents $a_1,a_2,a_3$ and six goods $g_1,\dots,g_6$. The goods $g_1,g_2,g_3$ are valued at $1$ by $a_1$ and at $0$ by $a_2$ and $a_3$. The goods $g_4,g_5,g_6$ are valued at $1$ by $a_2$ and $a_3$ and at $0$ by $a_1$. Any \PO{} allocation must assign $g_1,g_2,g_3$ to $a_1$ (giving it a utility of $3$) and allocate $g_4,g_5,g_6$ between $a_2$ and $a_3$. Either $a_2$ or $a_3$ receives at most one good, creating an \EQ{1} violation with $a_1$. Thus, an \EQ{1} and \PO{} allocation might fail to exist even under \emph{binary} valuations.
\label{eg:EQ1_PO_Nonexistence_binary_valuations}
\end{example}

Worse still, when the valuations can be zero-valued, determining whether there exists an \EQ{1}+\PO{} allocation is strongly \NPH{}. Similar hardness results hold for \EQx{}+\PO{} and \EQ{}+\PO{} allocations as well (\Cref{rem:EQ1+PO_EQx+PO_strong_NPhardness_GeneralVals}).

\begin{restatable}[\textbf{Hardness of \EQ{1} + \PO{}}]{theorem}{EQonePOStrongNPHardnessGeneralVals}
 \label{thm:EQ1+PO_StronglyNPHard_GeneralVals}
 Given any fair division instance with additive valuations, determining whether there exists an allocation that is equitable up to one good $(\EQ{1})$ and Pareto optimal $(\PO{})$ is strongly \NPhard{}.
\end{restatable}
\begin{proof}
We will show a reduction from \ThreePartition{}, which is known to be strongly \NPhard{} \citep{GJ79computers}. An instance of \ThreePartition{} consists of a set of $3r$ numbers $S = \{b_1,\dots,b_{3r}\}$  where $r \in \N$, and the goal is to find a partition of $S$ into $r$ subsets $S_1,\dots,S_r$ such that the sum of numbers in each subset is $T$, where $T \coloneqq \frac{1}{r} \sum_{a_i \in S} b_i$.\footnote{Note that we do not require $S_1,\dots,S_r$ to be of size three each; \ThreePartition{} remains strongly \NPhard{} even without this constraint.}

We will construct a fair division instance as follows: There are $r+1$ agents $a_1,\dots,a_{r+1}$ and $3r+2$ goods $g_1,\dots,g_{3r+2}$. For every $i \in [r]$ and $j \in [3r]$, agent $a_i$ values the good $g_j$ at $b_j$. The agents $a_1,\dots,a_r$ all value the goods $g_{3r+1}$ and $g_{3r+2}$ at $0$. Finally, the agent $a_{r+1}$ values $g_{3r+1}$ and $g_{3r+2}$ at $T$ each, and all other goods at $0$.

$(\Rightarrow)$ Suppose $S_1,\dots,S_r$ is a solution of \ThreePartition{}. Then, an \EQ{1} and \PO{} allocation $A = (A_1,\dots,A_{r+1})$ can be constructed as follows: For every $i \in [r]$, $A_i \coloneqq \{g_j : b_j \in S_i\}$, and $A_{r+1} \coloneqq \{g_{3r+1},g_{3r+2}\}$. Notice that $A$ is \EQ{1} because each of the agents $a_1,\dots,a_{r}$ has utility $T$, and the utility of the agent $a_{r+1}$ exceeds $T$ only by a single good $g_{3r+2}$. Furthermore, $A$ is \PO{} because each good is assigned to an agent with the highest valuation for it.

$(\Leftarrow)$ Now suppose that $A = (A_1,\dots,A_{r+1})$ is an \EQ{1} and \PO{} allocation. Since $A$ is \PO{}, it must assign $g_{3r+1}$ and $g_{3r+2}$ to $a_{r+1}$. Furthermore, since $A$ is \EQ{1}, each of the agents $a_1,\dots,a_{r}$ should have a utility of at least $T$ under $A$, i.e., for every $i \in [r]$, $v_i(A_i) \geq v_{r+1}(A_{r+1} \setminus \{g_{3r+2}\}) = T$. This induces a solution of the \ThreePartition{} instance.
\end{proof}

\begin{remark}[\textbf{Hardness of \EQx{}+\PO{}/\EQ{}+\PO{}}]
The reduction in \Cref{thm:EQ1+PO_StronglyNPHard_GeneralVals} can also be used to prove strong \NPhard{}ness of finding an \EQ{x}+\PO{} allocation (same construction works) or an \EQ{}+\PO{} allocation (if $a_{r+1}$ values $g_{3r+2}$ at $0$).
\label{rem:EQ1+PO_EQx+PO_strong_NPhardness_GeneralVals}
\end{remark}

Our next result shows that for the special case of \emph{binary} valuations (i.e., for all $i \in [n], j \in [m]$, $v_{i,j} \in \{0,1\}$), an \EQ{}+\PO{} allocation, if it exists, can be computed in polynomial time. Later, we will show similar tractability results for \EQ{1}+\PO{} and \EQx{}+\PO{} allocations (\Cref{thm:EQ1+EF1+PO_Polytime_BinaryVals}).

\begin{restatable}[\textbf{Algorithm for \EQ{}+\PO{} for binary valuations}]{theorem}{EQPOPolytimeBinaryVals}
 \label{thm:EQ+PO_Polytime_BinaryVals}
 There is a polynomial-time algorithm that given as input any fair division instance with additive and binary valuations, returns an allocation that is equitable $(\EQ{})$ and Pareto optimal $(\PO{})$ whenever such an allocation exists.
\end{restatable}
\begin{proof}
We will use a maximum flow algorithm. For binary valuations, an allocation is \PO{} if and only if it assigns each good to an agent that approves it. For an \EQ{} allocation $A$, we have $v_i(A_i) = v_k(A_k) = c$ (say) for every $i,k \in [n]$. Consider a bipartite graph $G = ([n] \cup [m],E)$ over the set of agents and goods with an edge $(i,j) \in E$ for every $i \in [n]$ and $j \in [m]$ such that $v_{i,j}=1$. For any fixed $c \in \N$, construct a flow network where the source node $S$ is connected to each agent node in $[n]$ with an edge of capacity $c$. Each node corresponding to a good in $[m]$ is connected to the sink node $T$ with an edge of capacity $1$. The edges between agents and goods are of capacity $1$. It is straightforward to check that there exists an \EQ{}+\PO{} allocation in the fair division instance (with common utility $c$) if and only if the above network admits a feasible flow of value $n \cdot c$. The desired algorithm simply iterates over all integral values of $c$ between $1$ and $\lfloor m/n \rfloor$.
\end{proof}

On the other hand, when all valuations are \emph{strictly positive} (i.e., $v_{i,j} > 0$ for all $i,j$), there always exists an allocation that is both equitable up to \emph{any} good and Pareto optimal.

\begin{restatable}[\textbf{Existence of \EQx{}+\PO{} for positive valuations}]{prop}{LeximinEQ1PO}
 \label{prop:Leximin_EQx+PO_Positive_valuations}
Given any fair division instance with additive and strictly positive valuations, an allocation that is equitable up to any good $(\EQx{})$ and Pareto optimal $(\PO{})$ always exists.
\end{restatable}
\begin{proof}(Sketch.)
We will show that any \Leximin{}-optimal allocation, say $A$, satisfies \EQx{} (Pareto optimality is easy to verify). Suppose, for contradiction, that there exist agents $i,k \in [n]$ and some good $j \in A_k$ such that $v_i(A_i) < v_k(A_k \setminus \{j\})$. Let $B$ be an allocation derived from $A$ by transferring the good $j$ from agent $k$ to agent $i$. Notice that under $B$, both agents $i$ and $k$ have strictly greater utility than $v_i(A_i)$, while all other agents have exactly the same utility as under $A$. Thus, $B$ is a `Leximin improvement' over $A$, which contradicts that $A$ is \Leximin{}-optimal.
\end{proof}

Although \Cref{prop:Leximin_EQx+PO_Positive_valuations} offers a strong existence result, it does not automatically provide a constructive procedure for finding such allocations. Indeed, computing a \Leximin{}-optimal allocation is known to be intractable \citep{BD05allocating,PR18almost}. Our next result (\Cref{thm:EQ1+PO_pseudopolynomial}) addresses this gap by providing a pseudopolynomial-time algorithm for finding an \EQ{1} and \PO{} allocation when the valuations are strictly positive.

\begin{restatable}[\textbf{Algorithm for \EQ{1}+\PO{} for positive valuations}]{theorem}{EQonePOPseudopolynomial}
 \label{thm:EQ1+PO_pseudopolynomial}
 Given any fair division instance $\I = \langle [n],[m],\V \rangle$ with additive and strictly positive valuations, an allocation that is equitable up to one good $(\EQ{1})$ and Pareto optimal $(\PO{})$ always exists and can be computed in $\O(\poly(m,n,v_{\max}))$ time, where $v_{\max} = \max_{i,j} v_{i,j}$.
\end{restatable}

In particular, when the valuations are polynomially bounded (i.e., for every $i \in [n]$ and $j \in [m]$, $v_{i,j} \leq \poly(m,n)$), our algorithm runs in \emph{polynomial time}. In contrast, computing a \Leximin{}-optimal allocation remains \NPH{} even under this restriction \citep{BD05allocating}. 

One might expect to prove \Cref{thm:EQ1+PO_pseudopolynomial} via a standard relax-and-round approach: Start with a \emph{fractional maximin allocation} (i.e., a fractional allocation that maximizes the minimum utility) followed by a rounding step. However, in \Cref{eg:Relax_And_Round_Fails}, we provide an instance where \emph{every} rounding of the fractional maximin solution fails to satisfy \EQ{1}. Therefore, the relax-and-round approach might be inadequate for finding \EQ{1}+\PO{} allocations.

Our proof of \Cref{thm:EQ1+PO_pseudopolynomial} is deferred to \Cref{subsec:EQ1+PO} but a brief idea is as follows: Our algorithm (Algorithm~\ref{alg:EQ1+PO}) uses the framework of \emph{Fisher markets} \citep{BS00compute}, which are well-studied models of a set of buyers spending their budgets of virtual money on utility-maximizing bundles of goods. Standard welfare theorems in economics guarantee that equilibrium (i.e., market clearing) outcomes in these markets are economically efficient. However, such outcomes could, in general, lead to \emph{fractional} allocations and be highly \emph{inequitable}. Our algorithm addresses the first challenge by starting with (and always maintaining) an \emph{integral} equilibrium of \emph{some} Fisher market. To meet the second challenge, our algorithm uses a combination of local search and price-rise routines to gradually move towards an \emph{approximately} equitable equilibrium. The analysis for achieving the desired running time and correctness guarantees is intricate, and involves a number of structural observations and potential function arguments.

Our techniques are inspired from a similar recent algorithm of \citet{BKV18Finding} for finding allocations that are envy-free up to one good (\EF{1}) and Pareto optimal (\PO{}). A key difference between the two algorithms lies in the way a local improvement is defined: For \citet{BKV18Finding}, a local improvement is defined in terms of equalizing the agents' \emph{spendings}, whereas for us, it pertains to equalizing the agents' \emph{utilities}. We believe that the latter approach is  more direct, and leads to a simpler algorithm and analysis. This distinction is also \emph{necessary}, because as we will show in \Cref{prop:Nonexistence_EQ1+EF1+PO}, an \EQ{1}+\EF{1}+\PO{} allocation might fail to exist even with strictly positive valuations. Therefore, any algorithm that is tailored to return an \EF{1} outcome---including the algorithm of \citet{BKV18Finding}---will invariably fail to find the desired \EQ{1}+\PO{} allocation, motivating the need for an alternative approach.

Given the success of market-based algorithms in finding \EQ{1}+\PO{} allocations, it is natural to ask whether these techniques can be extended to find an \EQx{}+\PO{} allocation. Unfortunately, this is where these techniques hit a roadblock. The problem stems from the fact that the market-based algorithm always outputs a \emph{fractionally Pareto optimal} (\fPO{}) allocation (refer to \Cref{subsec:EQ1+PO} for the definition), but there exist instances where no \EQx{} allocation satisfies \fPO{} (\Cref{subsec:NonExistence_EQx+fPO}). Whether an \EQx{}+\PO{} allocation can be computed in (pseudo-)polynomial time with strictly positive valuations is an intriguing question for future research.

\subsection{Equitability, Envy-Freeness and Pareto Optimality}

We will now consider all three notions---equitability, envy-freeness, and Pareto optimality---together. Recall from \Cref{prop:Leximin_EQx+PO_Positive_valuations} that for strictly positive valuations, an \EQ{1}+\PO{} (in fact, an \EQx{}+\PO{}) allocation is guaranteed to exist. It is also known that an \EF1{}+\PO{} allocation always exists. One might therefore ask whether an \EQ{1}+\EF1{}+\PO{} allocation also always exists. Our next result (\Cref{prop:Nonexistence_EQ1+EF1+PO}) dismisses that possibility.

\begin{restatable}[\textbf{Non-existence of \EQ{1}+\EF{1}+\PO{}}]{prop}{NonexistenceEQoneEFonePO}
 \label{prop:Nonexistence_EQ1+EF1+PO}
 There exists an instance with strictly positive valuations in which no allocation is simultaneously equitable up to one good $(\EQ{1})$, envy-free up to one good $(\EF{1})$ and Pareto optimal $(\PO{})$.
\end{restatable}
\begin{proof}
Fix some $n \geq 2$ and $0 < \eps < \frac{1}{2n+2}$. Consider an instance with $n+1$ agents $a_1,\dots,a_{n+1}$ and $3n+1$ goods $g_1,\dots,g_{3n+1}$. Each of $a_1,\dots,a_n$ values each of $g_1,\dots,g_{n-1}$ at $2$ and each of $g_n,\dots,g_{3n+1}$ at $\eps$. Agent $a_{n+1}$ values every good at $1$. By the pigeonhole principle for the goods $g_1,\dots,g_{n-1}$, some agent among $a_1,\dots,a_n$ must have utility at most $(2n+2)\eps < 1$. This means that $a_{n+1}$ can be assigned at most one good (otherwise \EQ{1} is violated). Therefore, if all the goods are allocated (which is a necessary condition for a \PO{} allocation), at least $3n$ goods must be assigned among $a_1,\dots,a_n$. This means that one of these agents gets at least three goods, creating an \EF{1} violation with $a_{n+1}$.
\end{proof}

\begin{remark}
\Cref{prop:Nonexistence_EQ1+EF1+PO} has several interesting implications. First, it shows that a Nash optimal allocation---which is guaranteed to be \EF{1} and \PO{} \citep{CKM+16unreasonable}---need not satisfy \EQ{1}. Similarly, the algorithm of \citet{BKV18Finding} for computing an \EF{1} and \PO{} allocation could also fail to return an \EQ{1} allocation. By contrast, our algorithm in \Cref{thm:EQ1+PO_pseudopolynomial} is guaranteed to find an \EQ{1} and \PO{} allocation. Finally, it shows that the \Leximin{}-optimal allocation---which is guaranteed to be \EQ{x} and \PO{} for strictly positive valuations (\Cref{prop:Leximin_EQx+PO_Positive_valuations})---need not be \EF{1}.
\label{rem:EQ1+EF1+PO_NonExistence_Implications}
\end{remark}

\paragraph{Comparison with cake-cutting}
It is worth comparing \Cref{prop:Nonexistence_EQ1+EF1+PO} with the corresponding results for \emph{divisible} goods (i.e., cake-cutting). \citet{BJK13n} have shown that there might not exist a division of the cake that simultaneously satisfies \EQ{}, \EF{}, and \PO{}. Our result in \Cref{prop:Nonexistence_EQ1+EF1+PO} shows an analogous impossibility for \emph{indivisible} goods. Interestingly, the impossibility for cake-cutting goes away when \PO{} is relaxed to \emph{completeness} (i.e., only requiring that the entire cake is allocated). Under this relaxation, it is known that a \emph{perfect} allocation of the cake exists \citep{A87splitting}.\footnote{An allocation $A$ is \emph{perfect} if for every $i,k \in [n]$, $v_i(A_k) = \frac{1}{n}$.} By contrast, for indivisible goods, the impossibility remains even when \PO{} is relaxed to completeness and \EF{1} is relaxed to proportionality up to one good (\Prop1).\footnote{An allocation $A$ is \emph{proportional} if for every $i \in [n]$, we have $v_i(A_i) \geq \frac{1}{n} \sum_{k \in [n]} v_i(A_k)$. An allocation $A$ is \emph{proportional up to one good}~\citep{CFS17fair} if for every $i \in [n]$, there exists a good $g$ such that $v_i(A_i \cup \{g\}) \geq \frac{1}{n} \sum_{k \in [n]} v_i(A_k)$.} Indeed, the proof of \Cref{prop:Nonexistence_EQ1+EF1+PO} works even under these relaxations. Moreover, the proof can be easily extended to show the non-existence of \EQ{k}, \Prop{$\ell$} and complete allocations for any constants $k,\ell \in \N$.

We now turn to the computational aspects of allocations with all three properties. Note that the allocation constructed in the proof of \Cref{thm:EQ1+PO_StronglyNPHard_GeneralVals} is envy-free. Therefore, from \Cref{thm:EQ1+PO_StronglyNPHard_GeneralVals,rem:EQ1+PO_EQx+PO_strong_NPhardness_GeneralVals}, we obtain strong \NPH{}ness of all combinations of the three properties.

\begin{restatable}[\textbf{Hardness of \EF{}+\EQ{}+\PO{}}]{corollary}{EFEQPOStrongNPHardnessGeneralVals}
 \label{cor:EF+EQ+PO_StronglyNPHard_GeneralVals}
 Let $X \in \{\EF{},\EFx{},\EF{1}\}$, $Y \in \{\EQ{},\EQx{},\EQ{1}\}$, and $Z = \PO{}$. Then, determining whether a given instance admits an allocation that is simultaneously $X$, $Y$, and $Z$ is strongly \NPhard{}.
\end{restatable}

The intractability in \Cref{cor:EF+EQ+PO_StronglyNPHard_GeneralVals} can, in certain cases, be alleviated when the valuations are restricted to be \emph{binary}. 
 We will start with an observation concerning \EQ{} and \PO{} allocations under this restriction.

\begin{restatable}{prop}{EQPOForBinaryIsAlsoEF}
 \label{prop:EQ+PO_BinaryVals_Is_EF}
 For binary valuations, an allocation that is equitable $(\EQ{})$ and Pareto optimal $(\PO{})$ is also envy-free $(\EF{})$.
\end{restatable}
\begin{proof}
Suppose each agent gets a utility $k$ under the said \EQ{} allocation. For binary valuations, \PO{} implies that each agent $i$ approves all the goods in its bundle. Furthermore, any other agent $j$ gets at most $k$ goods approved by $i$ (simply because agent $j$ gets exactly $k$ goods). Hence, the allocation is \EF{}.
\end{proof}

\begin{remark}
\Cref{prop:EQ+PO_BinaryVals_Is_EF} shows that for binary valuations, an \EQ{}+\PO{} allocation (if it exists) is, in fact, \EQ{}+\PO{}+\EF{} (hence also \EQ{}+\PO{}+\EFx{}/\EQ{}+\PO{}+\EF{1}). From \Cref{thm:EQ+PO_Polytime_BinaryVals}, we know that there is a polynomial-time algorithm for determining whether an instance with binary valuations admits an \EQ{}+\PO{} allocation. A similar implication therefore also holds for \EQ{}+\PO{}+\EF{}/\EQ{}+\PO{}+\EF{1}/\EQ{}+\PO{}+\EFx{} allocations.
\label{rem:EF_EQ_PO_Polytime_BinaryVals}
\end{remark}

\Cref{thm:EQ1+EF1+PO_Polytime_BinaryVals} shows that binary valuations are also useful when one considers the combination of \EQ{1}, \EF{1}, and \PO{}.

\begin{restatable}[\textbf{Algorithm for \EQ{1}+\EF{1}+\PO{} for binary valuations}]{theorem}{EQoneEFonePOPolytimeBinaryVals}
 \label{thm:EQ1+EF1+PO_Polytime_BinaryVals}
 There is a polynomial-time algorithm that given as input any fair division instance with additive and binary valuations, returns an allocation that is equitable up to one good $(\EQ{1})$, envy-free up to one good $(\EF{1})$, and Pareto optimal $(\PO{})$, whenever such an allocation exists.
\end{restatable}

The proof of \Cref{thm:EQ1+EF1+PO_Polytime_BinaryVals} is provided in \Cref{subsec:Proof_EQ1+EF1+PO_Polytime_BinaryVals}. The idea is to show that any \EQ{1}+\PO{} allocation, if it exists, is also Nash optimal. For binary valuations, all Nash optimal allocations induce identical utility profiles (up to renaming of agents). As a result, every Nash optimal allocation satisfies \EQ{1}. It is known that every Nash optimal allocation satisfies \EF{1} and \PO{} \citep{CKM+16unreasonable}. Moreover, for binary valuations, a Nash optimal allocation can be computed in polynomial time~\citep{DS15maximizing,BKV18greedy}. Therefore, determining the existence of an \EQ{1}+\EF{1}+\PO{} allocation reduces to checking whether an arbitrary Nash optimal allocation satisfies \EQ{1}, which can be done in polynomial time.
 
Notice that for binary valuations, a Pareto optimal allocation is \EF{1} if and only if it is \EFx{}, and is \EQ{1} if and only if it is \EQx{}. Therefore, when the valuations are binary, the above algorithm works for all combinations of $X$ + $Y$ + \PO{}, where $X \in \{\EFx{},\EF{1}\}$ and $Y \in \{\EQx{},\EQ{1}\}$.

We conclude this section by observing that some of the problems discussed in \Cref{cor:EF+EQ+PO_StronglyNPHard_GeneralVals} continue to be intractable even for binary valuations. This follows from a result of \citet{BL08efficiency}, who showed that finding an envy-free (\EF{}) and Pareto optimal (\PO{}) allocation under binary valuations is \NPC{} (refer to Proposition 21 in their paper).

\begin{restatable}[\citealp{BL08efficiency}]{prop}{EFPONPHardnessBinaryVals}
 \label{prop:EF_PO_NPHardness_BinaryVals}
 Given any fair division instance with additive and binary valuations, determining whether there exists an envy-free $(\EF{})$ and Pareto optimal $(\PO{})$ allocation is \NPC{}.
\end{restatable}

\begin{remark}
It is easy to verify that the allocation constructed in the reduction of \citet{BL08efficiency} is, without loss of generality, equitable up to one good (\EQ{1}). Therefore, for binary valuations, determining whether there exists an allocation that is \EF{} + \EQ{1}+ \PO{}/\EF{} + \EQx{}+ \PO{}  is \NPC{}.
\label{rem:EF_EQ1/x_PO_NPHardness_BinaryVals}
\end{remark}

\section{Experiments}
\label{sec:Experiments}
In this section, we compare the proposed and existing algorithms (in particular, \AlgEQonePO{}, MNW, and Leximin) in terms of how frequently they satisfy various fairness and efficiency properties in the real-world and synthetic datasets. 

For real-world preferences, we used the data obtained from the popular fair division website \emph{Spliddit}~\citep{GP15spliddit}. Out of the $2212$ instances in the Spliddit data, we used the $914$ instances that had strictly positive valuations and $m \geq n$. The instances have between $3$ and $9$ agents, and between $3$ and $29$ goods.\footnote{More than $80\%$ of the instances have three agents and six goods.} Users are restricted to normalized, integral valuations. For synthetic data, we generated $1000$ instances with $n=5$, $m=20$, and (strictly positive) valuations drawn i.i.d. from Dirichlet distribution. The concentration parameter for each item is set to $10$ to generate normalized valuations.\footnote{We normalize the valuations in the synthetic data to allow for a fair comparison with the Spliddit data, which has normalized valuations by design. We remark that all algorithms studied in this paper work even in the absence of this assumption.}

We consider the following combinations of fairness and efficiency properties: \EQ{}+\PO{}, \EQ{1}+\PO{}, \EQx{}+\PO{}, \EQ{1}+\EF{1}+\PO{}, and \EQx{}+\EFx{}+\PO{}. 
For each instance of the Spliddit and synthetic datasets, we check whether the property is satisfied by the output of \AlgEQonePO{}, MNW, and \Leximin{}.
\Cref{fig:Experimental_Results_Goods_Main} presents the relevant histograms.\footnote{All codes and synthetic data generation files are available at \url{https://github.com/sujoyksikdar/fairdivision}.}
Note that each of the algorithms we consider is Pareto optimal, so the histograms would be unaltered even if we did not assess \PO{}.

\begin{figure}
    \includegraphics[width=\linewidth]{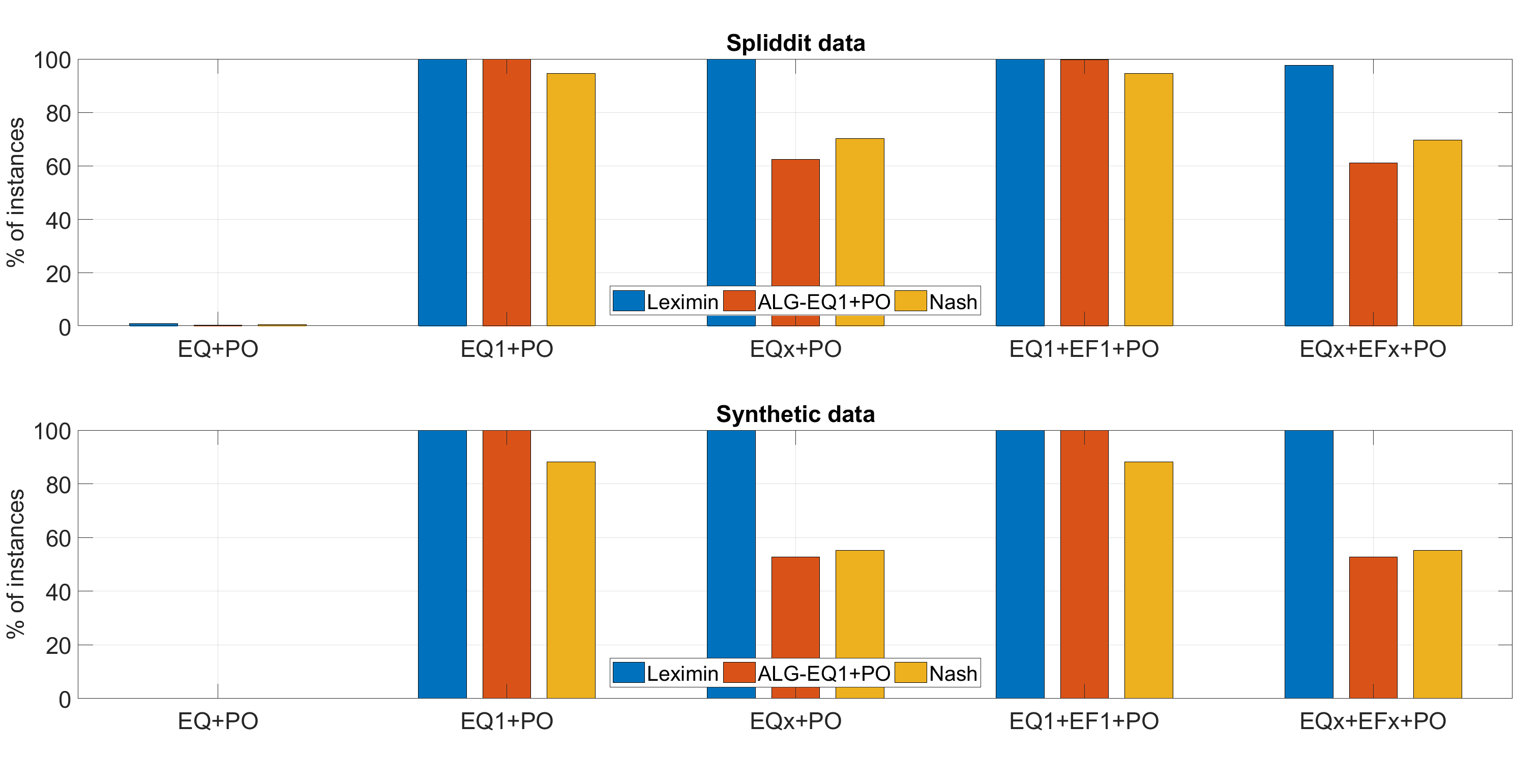}
\caption{Experimental results for Spliddit and synthetic datasets.}
\label{fig:Experimental_Results_Goods_Main}
\end{figure}

Not surprisingly, we see that very few instances permit a solution that is Pareto optimal and exactly equitable. Whenever such a solution exists, it is provably achieved by \Leximin{}, but this happens in only 1\% of Spliddit instances and none of the synthetic instances. For the \EQ{1} relaxation, we see that not only do \Leximin{} and \AlgEQonePO{} satisfy both \EQ{1} and \PO{}, but so does MNW on over $94\%$ of Spliddit instances (and over $88\%$ of synthetic instances). However, this trend changes when we consider \EQ{x}. \AlgEQonePO{}, despite being guaranteed to satisfy \EQ{1}, only satisfies \EQ{x} on $62\%$ of Spliddit instances (and $52\%$ of synthetic instances). A similar drop off is observed with MNW. Thus, for the purpose of achieving (approximately) equitable and Pareto optimal allocations, \Leximin{} is a clear winner.

We observe little change when, in addition to approximate equitability and Pareto optimality, we also require approximate envy-freeness. Indeed, in most cases, an allocation that is \EQ{1}+\PO{}/\EQ{x}+\PO{} is also \EF{1}/\EF{x}. It is interesting to note that while MNW---which is appealing from the perspective of achieving relaxed envy-freeness---quite often fails to satisfy \EQ{x}, \Leximin{} provably satisfies relaxed equitability while also achieving \EF{x} on a large fraction of instances.

\section{Discussion}
\label{sec:Discussion}
We studied equitable allocations of indivisible goods in conjunction with other well-known notions of fairness (envy-freeness) and economic efficiency (Pareto optimality), and provided a number of existential and computational results. In the appendix, we provide simulation results comparing the algorithms considered in Section~\ref{sec:Experiments} with respect to relaxations of envy-freeness (\Cref{subsec:Expt_Goods_EFx}). We also analyze \EQ{1} and \EQ{x} allocations from the perspective of approximating the optimal solutions to Max-Min Fairness, otherwise known as the \SantaClaus{} problem (\Cref{sec:maxmin}).

Our work reveals some intriguing similarities and differences between equitability and envy-freeness. In many places, our work parallels the existing literature on envy-freeness: We present Leximin as a canonical algorithm for \EQ{1}+\PO{}, just like MNW achieves \EF{1}+\PO{}. Also, our pseudopolynomial-time algorithm for \EQ{1}+\PO{} uses similar techniques to that of \citet{BKV18Finding} for \EF{1}+\PO{}. However, in other places, the differences are more pronounced. Most notably, \EQ{x} comes with a universal existence guarantee (often in conjunction with \PO{}), while the existence of \EF{x} allocations remains an open problem. Finally, exact equitability is a knife-edge property often hard to achieve in practice, unlike envy-freeness which is often satisfiable~\citep{DGK+14computational}.

Going forward, it would be very interesting to extend our results to the public decisions model of \citet{CFS17fair}. Extensions to models with additional feasibility constraints on the allocations \citep{BCE+17fair}, or settings with both goods and chores \citep{ACI18fair} will also be interesting.

\section*{Acknowledgments}
We are grateful to the anonymous IJCAI-19 reviewers for their helpful comments, and to Ariel Procaccia and Nisarg Shah for sharing with us the data from Spliddit. LX acknowledges NSF \#1453542 and \#1716333 for support.

\bibliographystyle{named}
\bibliography{References}

\normalsize
\clearpage
\newpage
\section{Appendix}
 \label{sec:Appendix}

\subsection{Proof of Theorem~\ref{thm:EQ1+PO_pseudopolynomial}}
\label{subsec:EQ1+PO}

Recall the statement of \Cref{thm:EQ1+PO_pseudopolynomial}.

\EQonePOPseudopolynomial*

The proof of \Cref{thm:EQ1+PO_pseudopolynomial} relies on the algorithm \AlgEQonePO{} (presented in Algorithm~\ref{alg:EQ1+PO}), and spans \Cref{subsec:EQ1+PO,subsec:Proof_RunningTime_ALG_EQ1+PO,subsec:Proof_RunningTime_Phase2,subsec:Proof_RunningTime_Phase3,subsec:Proof_Correctness_ALG_EQ1+PO_original_instance}.
We will start with some necessary definitions that will help us state \Cref{thm:epsEQ1+PO_pseudopolynomial}, of which \Cref{thm:EQ1+PO_pseudopolynomial} is a special case.

\paragraph{Fractional allocations}
A \emph{fractional allocation} $\x \in [0,1]^{n \times m}$ refers to a fractional assignment of the goods to the agents such that no more than one unit of any good is allocated, i.e., for every good $j \in [m]$, $\sum_{i \in [n]} x_{i,j} \leq 1$. We will use the term \emph{allocation} to refer to a discrete allocation and explicitly write \emph{fractional allocation} otherwise.

\paragraph{$\eps$-Pareto optimality}
Given any $\eps > 0$, $A$ is \emph{$\eps$-Pareto optimal} ($\eps$-\PO{}) if there does not exist an allocation $B$ such that $v_k(B_k) \geq (1+\eps) v_k(A_k)$ for every agent $k \in [n]$ with one of the inequalities being strict. 

\paragraph{Fractional Pareto optimality}
An allocation is \emph{fractionally Pareto optimal} (\fPO{}) if it not Pareto dominated by any fractional allocation. Thus, a fractionally Pareto optimal allocation is also Pareto optimal, but the converse is not necessarily true (\Cref{prop:Nonexistence_EQx+fPO}).

\paragraph{$\eps$-\EQ{1} allocation}
Given any $\eps > 0$, an allocation $A$ is $\eps$-\emph{equitable up to one good} ($\eps$-\EQ1) if for every pair of agents $i,k \in [n]$ such that $A_k \neq \emptyset$, there exists some good $j \in A_k$ such that $(1+\eps) v_i(A_i) \geq v_k(A_k \setminus \{j\})$.

\begin{restatable}{theorem}{epsEQonePOPseudopolynomial}
 \label{thm:epsEQ1+PO_pseudopolynomial}
 Given any fair division instance with additive and strictly positive valuations and any $\eps > 0$, an allocation that is $3\eps$-equitable up to one good $(3\eps\text{-}\EQ{1})$ and $\eps$-Pareto optimal $(\eps\text{-}\PO{})$ always exists and can be computed in $\O(\poly(m,n,\ln v_{\max},\nicefrac{1}{\eps}))$ time, where $v_{\max} = \max_{i,j} v_{i,j}$.
\end{restatable}

When $0 < \eps \leq \frac{1}{16m v_{\max}^4}$, we recover \Cref{thm:EQ1+PO_pseudopolynomial} as a special case of \Cref{thm:epsEQ1+PO_pseudopolynomial} (see \Cref{lem:Bound_On_Eps_For_Exact_PO,lem:Bound_On_eps_for_Exact_EQone}). 

The remainder of this section develops the necessary preliminaries that will enable us to present our algorithm (Algorithm~\ref{alg:EQ1+PO}) and the analysis of its running time (\Cref{lem:RunningTime_ALG_EQ1+PO}) and correctness (\Cref{lem:Correctness_ALG_EQ1+PO_original_instance}). The detailed proofs of these results are presented subsequently in \Cref{subsec:Proof_RunningTime_ALG_EQ1+PO,subsec:Proof_RunningTime_Phase2,subsec:Proof_RunningTime_Phase3,subsec:Proof_Correctness_ALG_EQ1+PO_original_instance}.

\subsubsection*{Market Preliminaries}

\paragraph{Fisher market}
 A Fisher market is an economic model that consists of a set of divisible goods and a set of agents (or buyers), each of whom is given a \emph{budget} (or \emph{endowment}) of virtual money \citep{BS00compute}. The agents can use the virtual money to purchase a utility-maximizing subset of the goods but do not derive any utility from the money itself. Formally, a Fisher market is given by a tuple $\M = \langle [n],[m],\V,e \rangle$ consisting of a set of $n$ \emph{agents} $[n] = \{1,2,\dots,n\}$, a set of $m$ divisible \emph{goods} $[m] = \{1,2,\dots,m\}$, a \emph{valuation profile} $\V = \{v_1,v_2,\dots,v_n\}$ and a vector of \emph{endowments} or \emph{budgets} $\e = (e_1,e_2,\dots,e_n)$. 

A \emph{market outcome} refers to a pair $(A,\p)$, where $A = (A_1,\dots,A_n)$ is a \emph{fractional allocation} of the $m$ goods, and $\p = (p_1,\dots,p_m)$ is a \emph{price vector} that associates a price $p_j \geq 0$ with every good $j \in [m]$. The \emph{spending} of agent $i$ under the market outcome $(A,\p)$ is given by $s_i = \sum_{j = 1}^m A_{i,j} p_j$. The \emph{utility} derived by the agent $i$ under $(A,\p)$ depends linearly on the valuations as $v_i(A_i) =  \sum_{j=1}^m A_{i,j} v_{i,j}$.

\paragraph{Induced fair division instance}
A Fisher market $\M = \langle [n],[m],\V,e \rangle$ naturally defines a fair division instance $\I = \langle [n],[m],\V \rangle$, which we will refer to as the \emph{induced fair division instance}. This correspondence between Fisher markets and the fair division problem allows us to extend the fairness and efficiency notions defined in \Cref{sec:Preliminaries} to Fisher markets. Thus, we will say that an allocation $A$ is equitable/envy-free/Pareto optimal for a \emph{market} $\M$ if it is equitable/envy-free/Pareto optimal for the induced fair division instance $\I$.

\paragraph{\MBB{} ratio and \MBB{} set}
Given a price vector $\p = (p_1,\dots,p_m)$, define the \emph{bang-per-buck} ratio of agent $i$ for good $j$ as $\alpha_{i,j} \coloneqq v_{i,j}/p_j$.\footnote{If $v_{i,j} = 0$ and $p_j = 0$, then $\alpha_{i,j} \coloneqq 0$.} The \emph{maximum bang-per-buck} \emph{ratio} (or \MBB{} ratio) of agent $i$ is $\alpha_i \coloneqq \max_j \alpha_{i,j}$. The \emph{maximum bang-per-buck} \emph{set} (or \MBB{} set) of agent $i$ is the set of all goods that maximize the bang-per-buck ratio for agent $i$ at the price vector $\p$, i.e., $\MBB_i \coloneqq \{j \in [m] : v_{i,j}/p_j = \alpha_i\}$.

A market outcome $(A,\p)$ constitutes an \emph{equilibrium} if it satisfies the following conditions:
	\begin{itemize}
	\item \emph{Market clearing}: Each good is either priced at zero or is completely allocated. That is, for every good $j \in [m]$, either $p_j = 0$ or $\sum_{i=1}^n A_{i,j} = 1$.
	\item \emph{Budget exhaustion}: Agents spend their budgets completely, i.e., $s_i = e_i$ for all $i \in [n]$.
	\item \emph{\MBB{} consistency}: Each agent's allocation is a subset of its $\MBB{}$ set. That is, for every agent $i \in [n]$ and every good $j \in [m]$, $A_{i,j} > 0 \implies j \in \MBB_i$. Note that \MBB{} consistency implies that every agent maximizes its utility at the given prices $\p$ under the budget constraints.
	\end{itemize}

\Cref{prop:FirstWelfareTheorem} presents the well-known \emph{first welfare theorem} for Fisher markets \citep[Chapter~16]{MWG+95microeconomic}.

\begin{restatable}{prop}{FirstWelfareTheorem}
 \label{prop:FirstWelfareTheorem}
 For a Fisher market with linear utilities, any equilibrium outcome is fractionally Pareto optimal $(\fPO{})$.
\end{restatable}

\paragraph{\MBB{}-allocation graph and alternating paths} 
Let $\M = \langle [n],[m],\V,e \rangle$ be Fisher market, and let $A$ and $\p$ denote an integral allocation and a price vector for $\M$, respectively. An \emph{\MBB{}-allocation graph} is an undirected bipartite graph $G$ with vertex set $[n] \cup [m]$ and an edge between agent $i \in [n]$ and good $j \in [m]$ if either $j \in A_i$ (called an \emph{allocation edge}) or $j \in \MBB_i$ (called an \emph{$\MBB$ edge}). Notice that if $A$ is \MBB{}-consistent (i.e., $j \in A_i \implies j \in \MBB_i$), then the allocation edges are a subset of \MBB{} edges.

For an \MBB{}-allocation graph, define an \emph{alternating path} $P = (i,j_1,i_1,j_2,i_2,\dots,i_{\ell - 1},j_\ell,k)$ from agent $i$ to agent $k$ (and involving the agents $i_1,i_2,\dots,i_{\ell-1}$ and the goods $j_1,j_2,\dots,j_\ell$) as a series of alternating \MBB{} and allocation edges such that $j_1 \in \MBB_i \cap A_{i_1}$, $j_2 \in \MBB_{i_1} \cap A_{i_2}$,$\dots$, $j_\ell \in \MBB_{i_{\ell - 1}} \cap A_k$. If such a path exists, we say that agent $k$ is \emph{reachable} from agent $i$ via an alternating path.\footnote{Note that no agent or good can repeat in an alternating path.} In this case, the \emph{length} of path $P$ is $2\ell$ since it consists of $\ell$ \MBB{} edges and $\ell$ allocation edges.

\paragraph{Reachability set} Let $G$ denote the \MBB{}-allocation graph of a Fisher market for the outcome $(A,\p)$. Fix a \emph{source} agent $i \in [n]$ in $G$. Define the \emph{level} of an agent $k \in [n]$ as half the length of the shortest alternating path from $i$ to $k$ if one exists (i.e., if $k$ is reachable from $i$), otherwise set the level of $k$ to be $n$. The level of the source agent $i$ is defined to be $0$. The \emph{reachability set} $\R_i$ of agent $i$ is defined as a level-wise collection of all agents that are reachable from $i$, i.e., $\R_i = (\R_i^{0},\R_i^{1},\R_i^{2},\dots,)$, where $\R_i^{\ell}$ denotes the set of agents that are at level $\ell$ with respect to agent $i$. Note that given an \MBB{}-allocation graph, a reachability set can be constructed in polynomial time via breadth-first search.

Given a reachability set $\R_i$, we can redefine an \emph{alternating path} as a set of alternating $\MBB{}$ and allocation edges \emph{connecting agents at a lower level to those at a higher level}. Formally, we will call a path $P = (i,j_1,i_1,j_2,i_2,\dots,i_{\ell - 1},j_\ell,k)$ \emph{alternating} if (1) $j_1 \in \MBB_i \cap A_{i_1}$, $j_2 \in \MBB_{i_1} \cap A_{i_2}$,$\dots$, $j_\ell \in \MBB_{i_{\ell - 1}} \cap A_k$, and (2) $\level(i) < \level(i_1) < \level(i_2) < \dots < \level(i_{\ell - 1}) < \level(k)$. Thus, an alternating path cannot have edges between agents at the same level.

\paragraph{Violators and path-violators} Given a Fisher market $\M = \langle [n],[m],\V,e \rangle$ and an allocation $A$, an agent $i \in [n]$ with the least utility among all the agents is called the \emph{reference agent}, i.e., $i \in \arg\min_{k \in [n]} v_k(A_k)$.\footnote{Ties are broken lexicographically.} An agent $k \in [n]$ is said to be a \emph{violator} if $A_k \neq \emptyset$ and for every good $j \in A_k$, we have that $v_k(A_k \setminus \{j\}) > v_i(A_i)$, where $i$ is the reference agent. Notice that the allocation $A$ is \EQ{1} if and only if there is no violator.

Given any $\eps > 0$, an agent $k \in [n]$ is an $\eps$-\emph{violator} if $A_k \neq \emptyset$ and for every good $j \in A_k$, we have $v_k(A_k \setminus \{j\}) > (1+\eps) v_i(A_i)$. Thus, an agent can be a violator without being an $\eps$-violator. An allocation $A$ is $\eps$-\EQ{1} if and only if there is no $\eps$-violator.

A closely related notion is that of a \emph{path-violator}. Let $i$ and $\R_i$ denote the reference agent and its reachability set, respectively. An agent $k \in \R_i$ is a \emph{path-violator} with respect to the alternating path $P = (i,j_1,i_1,j_2,i_2,\dots,i_{\ell - 1},j_\ell,k)$ if $v_k(A_k \setminus \{j_\ell\}) > v_i(A_i)$. Note that a path-violator (along a path $P$) need not be a violator as there might exist some good $j \in A_k$ not on the path $P$ such that $v_k(A_k \setminus \{j\}) \leq v_i(A_i)$. Finally, given any $\eps > 0$, an agent $k \in \R_i$ is an \emph{$\eps$-path-violator} with respect to the alternating path $P = (i,j_1,i_1,\dots,j_\ell,k)$ if $v_k(A_k \setminus \{j_\ell\}) > (1+\eps) v_i(A_i)$.

\paragraph{$\eps$-rounded instance}
Given any $\eps > 0$, an \emph{$\eps$-rounded instance} refers to a fair division instance $\langle [n], [m], \V \rangle$ in which the valuations are either zero or a non-negative integral power of $(1+\eps)$. That is, for every agent $i \in [n]$ and every good $j \in [m]$, we have $v_{i,j} \in \{0,(1+\eps)^t\}$ for some $t \in \N \cup \{0\}$. 

Given any instance $\I = \langle [n], [m], \V \rangle$, the \emph{$\eps$-rounded version} of $\I$ is an instance $\I' = \langle [n], [m], \W \rangle$ obtained by \emph{rounding up} the valuations in $\I$ to the nearest integral power of $(1+\eps)$. That is, the \emph{$\eps$-rounded version} of instance $\I = \langle [n], [m], \V \rangle$ is an $\eps$-rounded instance $\I' = \langle [n], [m], \W \rangle$ constructed as follows: For every agent $i \in [n]$ and every good $j \in [m]$, $ w_{i,j}:= (1+\eps)^{\lceil \log_{1+\eps}v_{i,j} \rceil}$ if $v_{i,j}> 0$, and $0$ otherwise. Notice that $v_{i,j} \leq w_{i,j} \leq (1+\eps)v_{i,j}$ for every agent $i$ and every good $j$. We will assume that the rounded valuations are also \emph{additive}, i.e., for any set of goods $S \subseteq [m]$, $w_i(S) \coloneqq  \sum_{j \in S} w_{i,j}$. 

\subsubsection*{Description of the Algorithm}

Given an input instance $\I = \langle [n],[m],\V \rangle$, we first construct its $\eps$-rounded version $\I' = \langle [n],[m],\W \rangle$, which is then provided as an input to \AlgEQonePO{} (Algorithm~\ref{alg:EQ1+PO}).

The algorithm consists of three phases. In Phase 1, each good is assigned to an agent with the highest valuation for it (Line~\ref{algline:Assignment_Phase1}). This ensures that the initial allocation is \emph{integral} as well as \emph{fractionally Pareto optimal} (\fPO{}).\footnote{Indeed, the said allocation is \MBB{}-consistent with respect to the prices in Line~\ref{algline:Prices_Phase1}, and is therefore an equilibrium outcome of a Fisher market in which each agent is provided a budget equal to its spending under the allocation. From \Cref{prop:FirstWelfareTheorem}, the allocation is \fPO{}.} (These two properties are always maintained by the algorithm.) If the allocation at the end of Phase 1 is $\eps$-\EQ{1} with respect to the rounded instance $\I'$, then the algorithm terminates and returns this allocation as the output (Line~\ref{algline:TerminatePhase1}). Otherwise, it proceeds to Phase 2.

The allocation at the start of Phase 2 is not $\eps$-\EQ{1}, so there must exist an $\eps$-violator. Starting from the level $\ell = 1$ (Line~\ref{algline:InitializeLevel_Phase2}), the algorithm now performs a level-by-level search for an $\eps$-violator in the reachability set of the reference agent (Line~\ref{algline:IfCondition_Phase2}). As soon as an $\eps$-violator, say $h$, is found (along some alternating path $P$), the algorithm performs a pairwise swap between $h$ and the agent that precedes it along $P$ (Line~\ref{algline:Swap_Phase2}). Since the swapped good is in the \MBB{} sets of both agents, the allocation continues to be \MBB{}-consistent after the swap. If, at any stage, the reference agent ceases to be the least-utility agent, Phase 2 restarts with the new reference agent (Line~\ref{algline:IdentityChange_Phase2}).

The above process continues until either the current allocation becomes $\eps$-\EQ{1} for the rounded instance $\I'$ (in which case the algorithm terminates and returns the current allocation as the output in Line~\ref{algline:TerminatePhase2}), or if no $\eps$-violator is reachable from the reference agent (Line~\ref{algline:WhileLoop_Phase2}). In the latter case, the algorithm proceeds to Phase 3.

Phase 3 involves uniformly raising the prices of all the \emph{reachable goods}, i.e., the set of all goods that are collectively owned by all agents that are reachable from the reference agent (Line~\ref{algline:PriceRise_Phase3}). The prices are raised until a previously non-reachable agent becomes reachable due to the appearance of a new \MBB{} edge (Line~\ref{algline:Start_Of_Phase_3}). The algorithm now switches back to Phase 2 to start a fresh search for an $\eps$-violator in the updated reachability set (Line~\ref{algline:GoBackToPhase2_Phase3}).

\paragraph{Comparison with the algorithm of \citet{BKV18Finding}}
As mentioned previously in \Cref{sec:Results}, our algorithm is inspired from the algorithm of \citet{BKV18Finding} for achieving envy-freeness up to one good (\EF{1}) together with Pareto optimality (\PO{}). At a high-level, both algorithms involve searching for a reachable violator (along an alternating path). If such an agent exists, then it loses a good through a pairwise swap. Otherwise, both algorithms use price-rise in order to discover a new \MBB{} edge to a previously unreachable agent. The main difference between the two algorithms is that \citet{BKV18Finding} define a violator in terms of excess \emph{spending}, whereas we define a violator in terms of excess \emph{utility}.\footnote{Formally, in the framework of \citet{BKV18Finding}, an agent $k \in [n]$ is an $\eps$-\emph{violator} if for every good $j \in A_k$, we have that $\p(A_k \setminus \{j\}) > (1+\eps) \p(A_i)$; here $\p(G) \coloneqq \sum_{g \in G} p_g$ is the sum of prices of all the goods in the bundle $G$, and $i$ is the least spender.} In other words, their algorithm performs local search in the space of spendings, whereas our algorithm does so in the space of utilities. As a result of this small but subtle difference, \citet{BKV18Finding} achieve an approximate equitability condition in terms of spendings (which they call price envy-freeness up to one good), whereas we are able to guarantee a similar property in terms of the utilities, which is precisely the desired \EQ{1} condition.

\paragraph{Analysis of the algorithm}
The running time and correctness of our algorithm are established by \Cref{lem:RunningTime_ALG_EQ1+PO} and \Cref{lem:Correctness_ALG_EQ1+PO_original_instance}, respectively, as stated below.

\begin{restatable}[\textbf{Running time}]{lemma}{RunningTimeALGEQonePO}
 \label{lem:RunningTime_ALG_EQ1+PO}
 Given as input any $\eps$-rounded instance with strictly positive valuations, \AlgEQonePO{} terminates in $\O(\poly(m,n,\ln v_{\max},\nicefrac{1}{\eps}))$ time steps, where $v_{\max} = \max_{i,j} v_{i,j}$.
\end{restatable}

The proof of \Cref{lem:RunningTime_ALG_EQ1+PO} appears in \Cref{subsec:Proof_RunningTime_ALG_EQ1+PO}.

\begin{restatable}[\textbf{Correctness}]{lemma}{CorrectnessALGEQonePOOriginal}
 \label{lem:Correctness_ALG_EQ1+PO_original_instance}
 Let $\I$ be any fair division instance with strictly positive valuations and $\I'$ be its $\eps$-rounded version for any given $\eps > 0$. Then, the allocation $A$ returned by \AlgEQonePO{} for the input $\I'$ is $3\eps$-\EQ{1} and $\eps$-\PO{} for $\I$. In addition, if $\eps \leq \frac{1}{16m v_{\max}^4}$, then $A$ is \EQ{1} and \PO{} for $\I$.
\end{restatable}

The proof of \Cref{lem:Correctness_ALG_EQ1+PO_original_instance} appears in \Cref{subsec:Proof_Correctness_ALG_EQ1+PO_original_instance}.

Notice that the running time guarantee in \Cref{lem:RunningTime_ALG_EQ1+PO} is stated in terms of \emph{time steps}. A time step refers to a single iteration of Phase 1, Phase 2, or Phase 3. Since each individual iteration requires polynomial time, it suffices to analyze the running time of the algorithm in terms of the \emph{number} of iterations of the three phases.\footnote{Indeed, an iteration of Phase 1 involves assigning each good to the agent with the highest valuation and setting its price. An iteration of Phase 2 involves the construction of the reachability set (say via breadth-first or depth-first search), followed by performing a level-wise search for an $\eps$-path-violator, followed by performing a swap operation. An iteration of Phase 3 involves scanning the set of reachable goods and setting an appropriate value of the price-rise factor $\Delta$. All of these operations can be carried out in $\O(\poly(m,n))$ time.} We will use the terms \emph{step}, \emph{time step}, and \emph{iteration} interchangeably.

\renewcommand{\floatpagefraction}{.95}
\begin{algorithm}
 \DontPrintSemicolon
 \linespread{1.2}\selectfont
 \KwIn{An $\eps$-rounded instance $\I' = \langle [n],[m],\W \rangle$.}
 \KwOut{An integral allocation $A$.}
 \BlankLine
 \Comment{Phase 1: Initialization}
 \BlankLine
 \tikzmk{A}
 $A \leftarrow $ a utilitarian welfare maximizing allocation (assign good $j \in [m]$ to agent $i$ if $i \in \arg\max_{k \in [n]} w_{k,j}$)\label{algline:Assignment_Phase1}\;
 $\p \leftarrow $ For every good $j \in [m]$, set $p_j = w_{i,j}$ if $j \in A_i$\label{algline:Prices_Phase1}\;
 \lIf{$A$ is $\eps$-\EQ{1} for $\I'$\label{algline:TerminatePhase1}}{\KwRet{$A$}}
 \nonl \tikzmk{B}
 \boxit{mygray}
 \BlankLine
 \Comment{Phase 2: Remove \EQ{1} violations among the reachable agents}
 \BlankLine
 \oldnl \tikzmk{A}
 $i \leftarrow $ reference agent in $A$ \Comment*[r]{tiebreak lexicographically}\label{algline:Refresh_ReferenceAgent}
 $\R_i \leftarrow $ Reachability set of $i$ under $(A,\p)$\;
 $\ell = 1$ \Comment*[r]{initialize the level}\label{algline:InitializeLevel_Phase2}
 \While{$\R_i^{\ell}$ is non-empty and $A$ is not $\eps$-\EQ{1}\label{algline:WhileLoop_Phase2}}{
 		\uIf{$h \in \R_i^{\ell}$ is an $\eps$-path-violator along the alternating path $P = (i,j_1,h_1,\dots,j_{\ell - 1},h_{\ell - 1},j,h)$\label{algline:IfCondition_Phase2}}{$A_h \leftarrow A_h \setminus \{j\}$ and $A_{h_{\ell - 1}} \leftarrow A_{h_{\ell - 1}} \cup \{j\}$ \Comment*[r]{swap $j$}\label{algline:Swap_Phase2}
 			Repeat Phase 2 starting from Line~\ref{algline:Refresh_ReferenceAgent}\label{algline:IdentityChange_Phase2}
 		}\Else{$\ell \leftarrow \ell + 1$ \Comment*[r]{Proceed to the next level}}
 }\label{algline:MoveToNextLevel_Phase2}
  \lIf{$A$ is $\eps$-\EQ{1} for $\I'$\label{algline:TerminatePhase2}}{\KwRet{$A$}}
 \nonl \tikzmk{B}
 \boxit{mygray}
 \BlankLine
 \Comment{Phase 3: Price-rise}
 \BlankLine
 \oldnl \tikzmk{A} 
 $\Delta \leftarrow \min\limits_{h \in \R_i, \, j \in [m] \setminus A_{\R_i}} \frac{\beta_h}{w_{h,j}/p_j}$, where $\beta_h$ is the \MBB{} ratio of $h$ (in $\I'$) and $A_{\R_i} \coloneqq \cup_{h \in \R_i} A_h$ is the set of reachable goods 
 \label{algline:Start_Of_Phase_3}
 \BlankLine
 \Comment{$\Delta $ is the smallest price-rise factor that makes a new agent reachable}
 \BlankLine
 \ForEach{good $j \in A_{\R_i}$}{
 	$p_j \leftarrow \Delta \cdot p_j$\Comment*[r]{uniformly raise the prices of reachable goods}
 \label{algline:PriceRise_Phase3}}
 Repeat Phase 2 starting from Line~\ref{algline:Refresh_ReferenceAgent}\label{algline:GoBackToPhase2_Phase3}\;
 \nonl \tikzmk{B}
 \boxit{mygray}
 \caption{\AlgEQonePO{}}
 \label{alg:EQ1+PO}
\end{algorithm}

We are now ready to prove \Cref{thm:EQ1+PO_pseudopolynomial}.

\EQonePOPseudopolynomial*

\begin{proof}
Fix $\eps = \frac{1}{16 m v_{\max}^4}$. The $\eps$-rounded version $\I'$ can be constructed in $\O( \poly(m,n,\ln v_{\max}) )$ time. We run the algorithm \AlgEQonePO{} on the input $\I'$. From \Cref{lem:RunningTime_ALG_EQ1+PO}, we know that the algorithm terminates in $\O(\poly(m,n,\ln v_{\max},\nicefrac{1}{\eps}))$ time. \Cref{lem:Correctness_ALG_EQ1+PO_original_instance} implies that $A$ is \EQ{1} and \PO{} for $\I$.
\end{proof}

\subsection{Proof of Lemma~\ref{lem:RunningTime_ALG_EQ1+PO}}
\label{subsec:Proof_RunningTime_ALG_EQ1+PO}

Recall the statement of \Cref{lem:RunningTime_ALG_EQ1+PO}.
\RunningTimeALGEQonePO*
\begin{proof}
The proof of \Cref{lem:RunningTime_ALG_EQ1+PO} follows immediately from \Cref{lem:RunningTime_Phase2,lem:RunningTime_Phase3}, which are stated below.
\end{proof}

\begin{restatable}{lemma}{RunningTimePhaseTwo}
 \label{lem:RunningTime_Phase2}
 There can be at most $\O( \poly(m,n,\nicefrac{1}{\eps}) \ln m v_{\max} )$ consecutive iterations of Phase 2 before a Phase 3 step occurs.
\end{restatable}

\begin{restatable}{lemma}{RunningTimePhaseThree}
 \label{lem:RunningTime_Phase3}
 There can be at most $\O( \poly(n,\nicefrac{1}{\eps})  \ln v_{\max} )$ Phase 3 steps during any execution of \AlgEQonePO{}.
\end{restatable}

The proofs of \Cref{lem:RunningTime_Phase2,lem:RunningTime_Phase3} are provided in \Cref{subsec:Proof_RunningTime_Phase2,subsec:Proof_RunningTime_Phase3}, respectively.

\subsection{Proof of Lemma~\ref{lem:RunningTime_Phase2}}
\label{subsec:Proof_RunningTime_Phase2}

The proof of \Cref{lem:RunningTime_Phase2} relies on several intermediate results (\Cref{lem:Bound_On_Consecutive_Swaps,lem:Reference_Utility_Nondecreasing,lem:Lower_Bound_On_Increase_In_Utility,lem:Bound_On_Identity_Changes}) that are stated below.

\begin{restatable}{lemma}{BoundOnConsecutiveSwaps}
 \label{lem:Bound_On_Consecutive_Swaps}
 There can be at most $\O( \poly(m,n) )$ consecutive swap operations in Phase 2 before either the identity of the reference agent changes or a Phase 3 step occurs.
\end{restatable}

The proof of \Cref{lem:Bound_On_Consecutive_Swaps} is identical to \citep[Lemma 13]{BKV18Finding} and is therefore omitted.

Throughout, we will use the phrase \emph{at time step $t$} to refer to the state of the algorithm at the beginning of the time step $t$. In addition, we will use $i_t$ and $A^t \coloneqq (A^t_1,\dots,A^t_n)$ to denote the reference agent and the allocation maintained by the algorithm at the beginning of time step $t$, respectively. Thus, for instance, the utility of the reference agent at time step $t$ is $w_{i_t}(A^t_{i_t})$.

\begin{restatable}{lemma}{ReferenceUtilityNondecreasing}
 \label{lem:Reference_Utility_Nondecreasing}
 The utility of the reference agent cannot decrease with time. That is, for any time step $t$, 
\begin{align*}
w_{i_t}(A^t_{i_t}) \leq w_{i_{t+1}}(A^{t+1}_{i_{t+1}}).
\end{align*}
\end{restatable}
\begin{proof}
The only way in which the utility of a reference agent can change is via a swap operation in Phase 2. By construction, a reference agent can never lose a good during a swap operation (though it can possibly receive a good). Therefore, the utility of a reference agent cannot decrease.
\end{proof}

\begin{restatable}{lemma}{LowerBoundOnIncreaseInUtility}
 \label{lem:Lower_Bound_On_Increase_In_Utility}
 Let $i$ be a fixed agent. Consider any set of consecutive Phase 2 steps during the execution of \AlgEQonePO{}. Suppose that $i$ turns from a reference to a non-reference agent during time step $t$. Let $t' > t$ be the first time step after $t$ at which $i$ once again becomes a reference agent. Then, either $A_i^t$ is a strict subset of $A_i^{t'}$ or $w_i(A_i^{t'}) > (1+\eps) w_i(A_i^t)$.
\end{restatable}
\begin{proof}
In order for a reference agent to turn into a non-reference agent, it must receive a good during a swap operation. That is, agent $i$ must receive a good at time $t$ and hence $A_i^t$ is a strict subset of $A_i^{t+1}$. If agent $i$ does not lose any good between $t+1$ and $t'$, then the claim follows. Therefore, for the rest of the proof, we will assume that agent $i$ loses at least one good between $t+1$ and $t'$.

Among all the time steps between $t+1$ and $t'$ at which agent $i$ loses a good, let $\tau$ be the last one. Let $i_\tau$ be the reference agent at time step $\tau$. Since the utility of the reference agent is non-decreasing with time (\Cref{lem:Reference_Utility_Nondecreasing}), we have that
\begin{align}
w_{i_\tau}(A_{i_\tau}^\tau) \geq w_i(A_i^t).
\label{eqn:LowerBoundOnIncreaseInUtility_1}
\end{align}
Let $g$ denote the good lost by agent $i$ at time step $\tau$. An agent that loses a good must be an $\eps$-path violator (with respect to an alternating path involving that good). Therefore,
\begin{align}
w_i(A_i^\tau \setminus \{g\}) > (1+\eps) w_{i_\tau}(A_{i_\tau}^\tau).
\label{eqn:LowerBoundOnIncreaseInUtility_2}
\end{align}
Since $i$ does not lose any good between $\tau$ and $t'$, we have
\begin{align}
w_i(A_i^{t'}) \geq w_i(A_i^{\tau+1}) = w_i(A_i^\tau \setminus \{g\}).
\label{eqn:LowerBoundOnIncreaseInUtility_3}
\end{align}
Combining \Cref{eqn:LowerBoundOnIncreaseInUtility_1,eqn:LowerBoundOnIncreaseInUtility_2,eqn:LowerBoundOnIncreaseInUtility_3} gives 
\begin{align*}
w_i(A_i^{t'}) > (1+\eps) w_i(A_i^t),
\end{align*}
 as desired.
\end{proof}

\begin{restatable}{lemma}{BoundOnIdentityChanges}
 \label{lem:Bound_On_Identity_Changes}
 There can be at most $\O( \poly(m,n,\nicefrac{1}{\eps}) \ln m v_{\max} )$ changes in the identity of the reference agent before a Phase 3 step occurs.
\end{restatable}
\begin{proof}
From \Cref{lem:Lower_Bound_On_Increase_In_Utility}, we know that each time the algorithm cycles back to a some agent $i$ as the reference agent, either the allocation of agent $i$ grows strictly by at least one good, or its utility increases by at least a multiplicative factor of $(1+\eps)$. By pigeonhole principle, after every $n$ consecutive changes in the identity of the reference agent, the algorithm must cycle back to some agent as the reference. Along with the fact that the utility of the reference agent is non-decreasing with time (\Cref{lem:Reference_Utility_Nondecreasing}), we get that after every $mn$ consecutive identity changes, the utility of the reference agent must grow multiplicatively by a factor of $(1+\eps)$. Since the utility of any agent can be at most $m w_{\max}$ (where $w_{\max} = \max_{i,j} w_{i,j}$), there can be at most $mn \log_{1+\eps} m w_{\max}$ changes in the identity of the reference agent during the execution of the algorithm. Furthermore, for $\eps$-rounded valuations, we have that $w_{\max} \leq (1+\eps) v_{\max}$. The stated bound now follows by observing that $\frac{1}{\ln (1+\eps)} \leq \frac{2}{\eps}$ for every $\eps \in (0,1)$.
\end{proof}

We are now ready to prove \Cref{lem:RunningTime_Phase2}.
\RunningTimePhaseTwo*

\begin{proof}
From \Cref{lem:Bound_On_Identity_Changes}, we know that there can be at most $\O( \poly(m,n,\nicefrac{1}{\eps}) \ln m v_{\max} )$ changes in the identity of the reference agent (in Phase 2) before a Phase 3 step occurs. Furthermore, \Cref{lem:Bound_On_Consecutive_Swaps} implies that there can be at most $\O( \poly(m,n) )$ swap operations between two consecutive identity changes or an identity change and a Phase 3 step. Combining these implications gives the desired bound. 
\end{proof}

\subsection{Proof of Lemma~\ref{lem:RunningTime_Phase3}}
\label{subsec:Proof_RunningTime_Phase3}

The proof of \Cref{lem:RunningTime_Phase3} relies on several intermediate results (\Cref{lem:Set_Of_Violators_Cannot_Grow,lem:Allocation_Of_Violators_Cannot_Grow,lem:Violators_Are_Not_Reachable,lem:MBB_Lower_Bound_Violator_Good,cor:MBB_Lower_Bound_Implication}) that are stated and proved below. It will be useful to define the set $E_t$ of all $\eps$-violators at time step $t$. That is,
$$E_t \coloneqq \{k \in [n] : w_k(A^t_k \setminus \{j\}) > (1+\eps) w_{i_t}(A^t_{i_t}) \, \forall j \in A^t_k \},$$
where $i_t$ is the reference agent at time step $t$.

Some of our proofs will require the following assumption:
\begin{assumption}
	At the end of Phase 1 of \AlgEQonePO{}, every agent is assigned at least one good.
\label{assumption:Each_agent_gets_a_good}
\end{assumption}
This assumption can be ensured via efficient preprocessing techniques similar to those used by \citet{BKV18Finding}. We refer the reader to Section B.1 of their paper for details.

\begin{restatable}{lemma}{SetOfViolatorsCannotGrow}
 \label{lem:Set_Of_Violators_Cannot_Grow}
 Let $t$ and $t'$ be two Phase 3 time steps such that $t < t'$. Then, $E_{t'} \subseteq E_t$.
\end{restatable}
\begin{proof}
It suffices to consider consecutive Phase 3 steps $t$ and $t'$ such that all intermediate time steps $t+1, t+2, \dots, t'-1$ occur in Phase 2. Suppose, for contradiction, that there exists some agent $k \in E_{t'} \setminus E_t$. Observe that a non-$\eps$-violator cannot turn into an $\eps$-violator in Phase 3 as the allocation of the goods remains fixed during price-rise. Therefore, the only way in which $k$ can turn into an $\eps$-violator is via a swap operation in Phase 2. In the rest of the proof, we will argue that if there is a swap operation at time step $\tau$ (where $t < \tau < t'$) that turns $k$ into an $\eps$-violator, then there is a subsequent swap operation at time step $\tau + 1$ that turns it back into a non-$\eps$-violator. This will provide the desired contradiction.

Suppose that agent $k$ is at level $\ell$ in the reachability set when it receives a good $g$ that turns it into an $\eps$-violator. Recall that a swap operation involves transferring a good from an agent at a higher level $\ell + 1$ to one at a lower level $\ell$. Furthermore, a swap involving an agent at level $\ell + 1$ happens only when no agent in the levels $1,2,\dots,\ell$ is an $\eps$-path violator. Therefore, agent $k$ cannot be an $\eps$-path violator just before the time step $\tau$. In other words, there must exist a good $g'$ on an alternating path from the reference agent $i_\tau$ to agent $k$ such that
\begin{align}
(1+\eps) w_{i_\tau}(A^\tau_{i_\tau}) \geq w_k(A^\tau_k \setminus \{g'\}).
\label{eqn:Set_Of_Violators_Cannot_Grow_1}
\end{align}
Since agent $k$ becomes an $\eps$-violator (and hence an $\eps$-path violator) after receiving the good $g$, we have
\begin{align*}
w_k(A^\tau_k \cup \{g\} \setminus \{g'\}) & > (1+\eps) w_{i_{\tau+1}}(A^{\tau+1}_{i_{\tau+1}}) = (1+\eps) w_{i_\tau}(A^\tau_{i_\tau}),
\end{align*}
where the equality follows from the observation that neither the identity nor the allocation of the reference agent changes during the above swap. Note that the swap involving $g$ does not affect the alternating path to agent $k$ that includes the good $g'$. This means that agent $k$ now becomes the \emph{only} $\eps$-path-violator at level $\ell$ or below. Therefore, in a subsequent swap operation at time step $\tau + 1$, the algorithm will take $g'$ away from agent $k$, resulting in a new bundle $A_k^{\tau+1} = A^\tau_k \cup \{g\} \setminus \{g'\}$. From \Cref{eqn:Set_Of_Violators_Cannot_Grow_1}, we get that agent $k$ is a non-$\eps$-violator up to the removal of the good $g$, as desired.
\end{proof}

\begin{restatable}{lemma}{AllocationOfViolatorsCannotGrow}
 \label{lem:Allocation_Of_Violators_Cannot_Grow}
 Let $t$ and $t'$ be two Phase 3 time steps such that $t < t'$. Then, for any $k \in E_{t'}$, $A^{t'}_k \subseteq A^t_k$.
\end{restatable}
\begin{proof}(Sketch.)
Suppose, for contradiction, that there exists a good $g \in A^{t'}_k \setminus A^t_k$. The only way in which agent $k$ could have acquired the good $g$ is via a swap operation at time step $\tau$ for some $t < \tau < t'$. Thus, agent $k$ cannot be an $\eps$-path-violator just before the time step $\tau$, and therefore also cannot be an $\eps$-violator. By an argument similar to that in the proof of \Cref{lem:Set_Of_Violators_Cannot_Grow}, it follows that agent $k$ cannot be an $\eps$-violator at time step $t'$, giving us the desired contradiction.
\end{proof}

\begin{restatable}{lemma}{ViolatorsAreNotReachable}
 \label{lem:Violators_Are_Not_Reachable}
 For any Phase 3 time step $t$, $E_t \cap \R_{i_t} = \emptyset$.
\end{restatable}
\begin{proof} 
Suppose, for contradiction, that there exists some $k \in E_t \cap \R_{i_t}$ at time step $t$, i.e., $k$ is an $\eps$-violator that is reachable (via some alternating path). Then, agent $k$ must also be an $\eps$-path violator, implying that the algorithm continues to be in Phase 2 at time step $t$ and therefore cannot enter Phase 3.
\end{proof}

\begin{restatable}{lemma}{MBBLowerBound}
 \label{lem:MBB_Lower_Bound_Violator_Good}
 Let $t$ be a Phase 3 time step. Then, there exists an $\eps$-violator $k \in E_t$ and a good $j \in A_k^t$ such that for every agent $i \in [n]$, $\beta_i^t \geq \nicefrac{w_{i,j}}{w_{k,j}}$, where $\beta_i^t$ is the \MBB{} ratio of agent $i$ at time step $t$.
\end{restatable}
\begin{proof}
Note that the algorithm enters Phase 3 at time step $t$ only if the current allocation $A^t$ is not $\eps$-\EQ{1}. Thus, there must exist an $\eps$-violator agent $k \in E_t$. Fix any good $j \in A^t_k$ (this is well-defined since $A^t_k \neq \emptyset$). From \Cref{lem:Set_Of_Violators_Cannot_Grow,lem:Allocation_Of_Violators_Cannot_Grow}, we know that $k \in E_{\tau}$ and $j \in A^{\tau}_k$ for all Phase 3 time steps $\tau < t$. Additionally, for every Phase 3 time step $\tau$ preceding the time step $t$, we know from \Cref{lem:Violators_Are_Not_Reachable} that $k \notin R_{i_{\tau}}$. In other words, the agent $k$ never experiences a price-rise between the start of the algorithm and the time step $t$. As a result, the \MBB{} ratio of agent $k$ at time step $t$ is the same as that at the moment of the first price-rise, i.e., $\beta^t_k = \beta^{t_1}_k$, where $t_1$ denotes the earliest Phase 3 time step. Furthermore, since the \MBB{} ratios of all agents remain unchanged during Phase 2, we must have that $\beta^{t_1}_k = 1$ (this follows from the way we set the initial prices in Phase 1), and thus also $\beta^t_k = 1$. By a similar argument, the good $j$ does not experience a price-rise between the start of the algorithm and the time step $t$. Therefore, $p^t_j = p^{t_1}_j$. Since the allocation maintained by the algorithm is always \MBB{}-consistent, we get that $p^{t_1}_j = w_{k,j}$. The claim now follows by noticing that each agent's \MBB{} ratio is at least its bang-per-buck ratio for the good $j$.
\end{proof}

\begin{restatable}{corollary}{MBBLowerBoundImplication}
 \label{cor:MBB_Lower_Bound_Implication}
 Let $t$ be a Phase 3 time step. Then, for every agent $i \in [n]$, we have $\beta_i^t \geq \frac{1}{w_{\max}}$, where $\beta_i^t$ is the \MBB{} ratio of agent $i$ at time step $t$, where $w_{\max} = \max_{i,j} w_{i,j}$.
\end{restatable}
\begin{proof}
By assumption, all valuations in the original instance $\I$ are strictly positive and integral. This means that in the $\eps$-rounded version $\I'$, for every $i \in [n]$ and $j \in [m]$, we have $1 \leq v_{i,j} \leq w_{i,j} \leq w_{\max}$. Using these inequalities in the bound from \Cref{lem:MBB_Lower_Bound_Violator_Good} gives the desired claim.
\end{proof}

We are now ready to prove \Cref{lem:RunningTime_Phase3}.
\RunningTimePhaseThree*

\begin{proof}
The proof uses a potential function argument. For any Phase 3 time step $t$, we define a potential 
\begin{align*}
\Phi^t \coloneqq \sum_{i \in [n]} \log_{1+\eps} \beta_i^t,
\end{align*}
where $\beta_i^t \coloneqq \max_{j \in [m]} \nicefrac{w_{i,j}}{p^t_j}$ is the \MBB{} ratio of agent $i$ and $p^t_j$ is the price of good $j$ at time step $t$. 

In Phase 1, the price of every good is set to be the highest valuation for that good. Along with Assumption~\ref{assumption:Each_agent_gets_a_good}, this implies that at the end of Phase 1, the \MBB{} ratio of every agent equals $1$. Since Phase 2 does not affect the prices, the \MBB{} ratio of every agent at the time of the earliest price-rise also equals $1$. Thus, the initial value of the potential $\Phi^1$ is $0$.

We will now argue that each time the algorithm performs a price-rise, the potential must decrease by at least $1$ (i.e., for any two Phase 3 steps $t$ and $t'$ such that $t < t'$, $\Phi^t - \Phi^{t'} \geq 1$). Recall that the algorithm never decreases the price of any good. Therefore, all bang-per-buck ratios (and hence all \MBB{} ratios) are non-increasing with time. Thus, $\Phi^t \leq 0$ for all time steps $t \in \{1,2,\dots\}$. In addition, each time the algorithm performs a price-rise, the \MBB{} ratio of some agent \emph{strictly} decreases (because a new good gets added to the \MBB{} set of some agent).

We will argue that the (multiplicative) drop in \MBB{} ratio is always by a positive integral power of $(1+\eps)$. Indeed, by assumption, all valuations are integral powers of $(1+\eps)$. We already observed earlier that all \MBB{} ratios at the end of Phase 1 are equal to $1$, which means that all initial prices must also be integral powers of $(1+\eps)$. Furthermore, the price-rise factor $\Delta$ is a ratio of bang-per-buck ratios, and is therefore also an integral power of $(1+\eps)$. So, whenever the \MBB{} ratio of some agent strictly decreases, it must be by an integral power of $(1+\eps)$. This means that after every price-rise in Phase 3, the potential must decrease by at least $1$.

All that remains to be shown is a lower bound on the potential $\Phi^t$. From \Cref{cor:MBB_Lower_Bound_Implication}, we know that for every Phase 3 time step $t$, we have $\beta^t_i \geq \frac{1}{w_{\max}}$, and consequently, $\Phi^t \geq - n  \log_{1+\eps} {w_{\max}}$. Since the potential decreases by at least $1$ between any consecutive price-rises, the overall number of Phase 3 time steps can be at most $n \log_{1+\eps} w_{\max}$. For $\eps$-rounded valuations, we have $w_{\max} \leq (1+\eps) v_{\max}$, and therefore $n \log_{1+\eps} w_{\max} = n + n \log_{1+\eps} v_{\max}$. The stated bound now follows by observing that $\frac{1}{\ln (1+\eps)} \leq \frac{2}{\eps}$ for every $\eps \in (0,1)$.
\end{proof}

\subsection{Proof of Lemma~\ref{lem:Correctness_ALG_EQ1+PO_original_instance}}
\label{subsec:Proof_Correctness_ALG_EQ1+PO_original_instance}

The proof of \Cref{lem:Correctness_ALG_EQ1+PO_original_instance} relies on several intermediate results (\Cref{lem:Correctness_ALG_EQ1+fPO_rounded_instance,lem:fPO_Rounded_eps_PO_Original,lem:Bound_On_Eps_For_Exact_PO,lem:epsEQ1_Rounded_3epsEQ1_Original,lem:Bound_On_eps_for_Exact_EQone}) as stated below.

\begin{restatable}{lemma}{CorrectnessALGEQonefPORounded}
 \label{lem:Correctness_ALG_EQ1+fPO_rounded_instance}
 Given as input any $\eps$-rounded instance $\I'$ with strictly positive valuations, the allocation $A$ returned by \AlgEQonePO{} is $\eps$-\EQ{1} and \fPO{} for $\I'$.
\end{restatable}
\begin{proof}
From \Cref{lem:RunningTime_ALG_EQ1+PO}, we know that \AlgEQonePO{} is guaranteed to terminate. Furthermore, the algorithm can only terminate in Lines~\ref{algline:TerminatePhase1} or \ref{algline:TerminatePhase2}. In both cases, the allocation $A$ returned by the algorithm is guaranteed to be $\eps$-\EQ{1} with respect to the input instance $\I'$.

To see why $A$ is \fPO{}, note that \AlgEQonePO{} always maintains an \MBB{}-consistent allocation (with respect to the current prices). Define a Fisher market where each agent is assigned a budget equal to its spending under $A$. Then, the outcome $(A,\p)$ satisfies the equilibrium conditions for this market. Therefore, from \Cref{prop:FirstWelfareTheorem}, $A$ is \fPO{}.
\end{proof}

\begin{restatable}{lemma}{fPOForRoundedIsEpsPOForOriginal}
 \label{lem:fPO_Rounded_eps_PO_Original}
 Let $\I$ be any fair division instance and $\I'$ be its $\eps$-rounded version for any given $\eps > 0$. Then, an allocation $A$ that is \fPO{} for $\I'$ is $\eps$-\PO{} for $\I$. 
\end{restatable}
\begin{proof}
Suppose, for contradiction, that $A$ is $\eps$-Pareto dominated in $\I$ by an allocation $B$, i.e., $v_k(B_k) \geq (1+\eps) v_k(A_k)$ for every agent $k \in [n]$ and $v_i(B_i) > (1+\eps)  v_i(A_i)$ for some agent $i \in [n]$. Since $\I'$ is an $\eps$-rounded version of $\I$, we have that $v_{i,j} \leq w_{i,j} \leq (1+\eps)v_{i,j}$ for every agent $i$ and every good $j$. Using this bound and the additivity of valuations, we get that $w_k(B_k) \geq w_k(A_k)$ for every agent $k \in [n]$ and $w_i(B_i) > w_i(A_i)$ for some agent $i \in [n]$. Thus, $B$ Pareto dominates $A$ in the instance $\I'$, which is a contradiction since $A$ is \fPO{} (hence \PO{}) for $\I'$.
\end{proof}

\begin{restatable}{lemma}{BoundOnPrice}
 \label{lem:Bound_On_Price}
 Let $\p$ denote the price-vector right after the termination of \AlgEQonePO{}. Then, for every good $j \in [m]$, we have $p_j \leq w_{\max}^4$.
\end{restatable}
\begin{proof}
Recall that the prices can change only during Phase 3. From \Cref{cor:MBB_Lower_Bound_Implication}, we know that \emph{before} the start of any Phase 3 time step $t$, we have $\beta_i^t \geq \frac{1}{w_{\max}}$ for every $i \in [n]$. Furthermore, at any time step $t$, each good $j$ is assigned to an agent $i$ such that $j \in \MBB_i$. Therefore, for every good $j \in A_i$, we have $p_j^t = \frac{w_{i,j}}{\beta_i^t} \leq w_{\max}^2$, where $p_j^t$ denotes the price of good $j$ at time step $t$. In particular, if $t_1,\dots,t_N$ denote the Phase 3 time steps during the execution of the algorithm, then just before the final price-rise at $t_N$, we have that $p_j^{t_N} \leq w_{\max}^2$ for every good $j \in [m]$.

Let $\Delta_N$ denote the (multiplicative) price-rise factor at time step $t_N$. By definition, $\Delta_N \leq \frac{\beta_h^{t_N}}{w_{h,j}/p_j^{t_N}}$ for every agent $h \in \R_i$ and every good $j \in [m] \setminus A_{\R_i}$; here, $\beta_h^{t_N}$ is the \MBB{} ratio of agent $h$ at time step $t_N$ and $A_{\R_i} \coloneqq \cup_{h \in \R_i} A_h$ is the set of reachable goods. Since the \MBB{} ratios are non-increasing with time, we have that $\beta_h^{t_N} \leq 1$ for all agents $h \in [n]$. Additionally, since all valuations in the original instance $\I$ are strictly positive integers and $\I'$ is its $\eps$-rounded version, we have that $w_{i,j} \geq v_{i,j} \geq 1$. The above observations give us that $\Delta_N \leq p_j^{t_N}$. The final price of every good is, therefore, at most $\Delta_N \cdot p_j^{t_N} \leq ( p_j^{t_N} )^2 \leq w_{\max}^4$.
\end{proof}

\begin{restatable}{lemma}{BoundOnEpsForExactPO}
 \label{lem:Bound_On_Eps_For_Exact_PO}
 Let $\I$ be any fair division instance and $\I'$ be its $\eps$-rounded version for any $0 < \eps \leq \frac{1}{16 m v_{\max}^4}$. Let $A$ be the allocation returned by \AlgEQonePO{} for the input instance $\I'$. Then, $A$ is \PO{} for $\I$.
\end{restatable}
\begin{proof}
Suppose, for contradiction, that the allocation $A$ is Pareto dominated by an allocation $B$ in the instance $\I$. That is, $v_k(B_k) \geq v_k(A_k)$ for every agent $k \in [n]$ and $v_i(B_i) > v_i(A_i)$ for some agent $i \in [n]$. Since the valuations in $\I$ are integral, we have $v_i(B_i) \geq v_i(A_i) + 1$. 

Let $\p$ denote the price-vector right after the termination of \AlgEQonePO{}. Let $\alpha_k$ and $\beta_k$ denote the \MBB{} ratios (with respect to $\p$) of agent $k \in [n]$ in $\I$ and $\I'$, respectively. Since $\I'$ is an $\eps$-rounded version of $\I$, we have $v_{k,j} \leq w_{k,j} \leq (1+\eps)v_{k,j}$ for every agent $k \in [n]$ and every good $j \in [m]$. Thus,
$\alpha_k = \max_{j \in [m]} v_{k,j}/p_j \leq \max_{j \in [m]} w_{k,j}/p_j = \beta_k.$
Therefore,
\begin{alignat}{2}
 \frac{v_k(A_k)}{\alpha_k} & \geq \frac{w_k(A_k)}{(1+\eps) \alpha_k} & \qquad \text{(since $\I'$ is $\eps$-rounded)} & \nonumber \\
 & \geq \frac{w_k(A_k)}{(1+\eps) \beta_k} & \qquad \text{(since $\alpha_k \leq \beta_k$)} & \nonumber \\
 & = \frac{ \p(A_k) } {1+\eps} & \text{(A is $\MBB$-consistent in $\I'$)}. & \label{eqn:Small_delta_PO_temp1}
\end{alignat}
Now consider the allocation $B$. By definition of \MBB{} ratio, we have that $\alpha_k \p(B_k) \geq v_k(B_k)$ for every $k \in [n]$. Since $B$ Pareto dominates $A$ in $\I$, we have $\alpha_k \p(B_k) \geq v_k(A_k)$. Along with \Cref{eqn:Small_delta_PO_temp1}, this gives
\begin{equation}
\p(B_k) \geq \frac{\p(A_k)}{1+\eps}.
\label{eqn:Small_delta_PO_temp2}
\end{equation}
By a similar reasoning for agent $i$, we get
\begin{equation}
\p(B_i) \geq \frac{\p(A_i)}{1+\eps} + \frac{1}{\alpha_i}.
\label{eqn:Small_delta_PO_temp3}
\end{equation}
The combined spending over all the goods is given by
\begin{align*}
 \p([m]) & = \sum_{k \in [n]} \p(B_k) & \text{(since all the goods are allocated under $B$)}\\
  & = \p(B_i) + \sum_{k \in [n] \setminus \{i\}} \p(B_k) \\
  & \geq \frac{\p(A_i)}{1+\eps} + \frac{1}{\alpha_i} + \sum_{k \in [n] \setminus \{i\}} \frac{\p(A_k)}{1+\eps} & (\text{from \Cref{eqn:Small_delta_PO_temp2,eqn:Small_delta_PO_temp3}})\\
  & = \frac{\p([m])}{1+\eps} + \frac{1}{\alpha_i}  & \text{(since all the goods are allocated under $A$)}.
\end{align*}
Further simplification gives
\begin{equation*}
1 \leq \eps \left( \p([m]) \alpha_i - 1 \right) < \eps \p([m]) \alpha_i \leq \eps \p([m]),
\end{equation*}
where the first inequality follows from simplifying the preceding relation, the second inequality follows from $\eps > 0$, and the third inequality uses the fact that $\alpha_i \leq 1$ for every agent $i$.\footnote{The reason we have $\alpha_i \leq 1$ is because all the \MBB{} ratios (for $\I'$) are initially equal to $1$ (this is discussed in the proof of \Cref{lem:RunningTime_Phase3}). Since the algorithm never decreases the price of any good, the final \MBB{} ratios are all at most $1$, i.e., $\beta_i \leq 1$ for all $i \in [n]$. Since $\alpha_i \leq \beta_i$, the relation follows.} From \Cref{lem:Bound_On_Price}, we know that $\p([m[) \leq m \cdot \max_{j \in [m]} p_j \leq m w_{\max}^4$. This requires that $\eps > \frac{1}{m w_{\max}^4}$. For $\eps$-rounded valuations, this means that $\eps (1+\eps)^4 > \frac{1}{m v_{\max}^4}$. 

Recall that $\eps \leq \frac{1}{16 m v_{\max}^4}$. Using this bound in the above expression for $\eps$ while keeping the $(1+\eps)$ term intact, we get that $(1+\eps)^4 > 16$, or, equivalently, $\eps > 1$, which is a contradiction. Hence, $A$ must be \PO{} for $\I$.
\end{proof}

\begin{restatable}{lemma}{EpsEQ1RoundedThreeEpsEQ1Original}
 \label{lem:epsEQ1_Rounded_3epsEQ1_Original}
 Let $\I$ be any fair division instance and $\I'$ be its $\eps$-rounded version for any given $\eps > 0$. Then, an allocation $A$ that is $\eps$-\EQ{1} for $\I'$ is $3\eps$-\EQ{1} for $\I$.
\end{restatable}
\begin{proof}
Since $A$ is $\eps$-\EQ{1} for $\I'$, we have that for every pair of agents $i,k \in [n]$ such that $A_k \neq \emptyset$, there exists a good $j \in A_k$ such that $(1+\eps) w_i(A_i) \geq w_k(A_k \setminus \{j\})$. Furthermore, since $\I'$ is an $\eps$-rounded version of $\I$, we have that $v_{i,j} \leq w_{i,j} \leq (1+\eps)v_{i,j}$ for every agent $i$ and every good $j$. Using this relation along with the additivity of valuations (in $\I'$), we get that for every pair of agents $i,k \in [n]$, there exists a good $j \in A_k$ such that $(1+\eps)^2 v_i(A_i) \geq v_k(A_k \setminus \{j\})$. Since $\eps < 1$, we get that there exists a good $j \in A_k$ such that $(1+3\eps) v_i(A_i) \geq v_k(A_k \setminus \{j\})$, implying that $A$ is $3\eps$-\EQ{1} for $\I$.
\end{proof}

\begin{restatable}{lemma}{BoundOnEpsforExactEQone}
 \label{lem:Bound_On_eps_for_Exact_EQone}
 Given any fair division instance $\I$ and any $0 \leq \eps \leq \frac{1}{6 m v_{\max}}$, an allocation $A$ is $3\eps$-\EQ{1} for $\I$ if and only if it is \EQ{1} for $\I$.
\end{restatable}
\begin{proof}
Since $A$ is $3\eps$-\EQ{1}, we have that for every pair of agents $i,k \in [n]$ such that $A_k \neq \emptyset$, there exists a good $j \in A_k$ such that $(1+3\eps) v_i(A_i) \geq v_k(A_k \setminus \{j\})$. Using the bound $\eps \leq \frac{1}{6 m v_{\max}}$, we get that there exists a good $j \in A_k$ such that $v_k(A_k \setminus \{j\}) - v_i(A_i) \leq \frac{1}{2}$. Since the valuations are integral, this implies that there exists a good $j \in A_k$ such that $v_k(A_k \setminus \{j\}) - v_i(A_i) \leq 0$, which is the \EQ{1} condition.
\end{proof}

We are now ready to prove \Cref{lem:Correctness_ALG_EQ1+PO_original_instance}.
\CorrectnessALGEQonePOOriginal*

\begin{proof} 
The allocation $A$ returned by \AlgEQonePO{} is guaranteed to be $\eps$-\EQ{1} and \fPO{} with respect to the input instance $\I'$ (\Cref{lem:Correctness_ALG_EQ1+fPO_rounded_instance}).  \Cref{lem:epsEQ1_Rounded_3epsEQ1_Original,lem:fPO_Rounded_eps_PO_Original} together imply that $A$ is $3\eps$-\EQ{1} and $\eps$-\PO{} for $\I$. Furthermore, if $\eps \leq \frac{1}{16m v_{\max}^4}$, then the bounds in \Cref{lem:Bound_On_Eps_For_Exact_PO,lem:Bound_On_eps_for_Exact_EQone} are satisfied, which implies that $A$ is \EQ{1} and \PO{} for $\I$.
\end{proof}

\subsection{\EQx{}+\fPO{} might fail to exist}
\label{subsec:NonExistence_EQx+fPO}

\begin{restatable}[\textbf{Non-existence of \EQx+\fPO{}}]{prop}{NonexistenceEQxfPO}
 \label{prop:Nonexistence_EQx+fPO}
 Given any $\eps > 0$, there exists an $\eps$-rounded instance with strictly positive valuations in which no allocation is simultaneously equitable up to any good $(\EQx{})$ and fractionally Pareto optimal $(\fPO{})$.
\end{restatable}
\begin{proof}
Consider the following instance with two agents and three goods:
\begin{table}[ht]
\centering
\begin{tabular}{ c|ccc }
	& $g_1$ & $g_2$ & $g_3$\\ \hline
  $a_1$ & $1+\eps$ & $(1+\eps)^{10}$ & $1$\\
  $a_2$ & $1$ & $(1+\eps)^{10}$ & $1+\eps$
\end{tabular}
\end{table}

Observe that all valuations are non-negative integral powers of $(1+\eps)$, and therefore satisfy the $\eps$-rounded property. This instance has two \EQx{} allocations $A$ and $B$ given by $A_1 = \{g_2\}, A_2 = \{g_1,g_3\}$ and $B_1 = \{g_1,g_3\}, B_2 = \{g_2\}$. The utilities of the agents under $A$ are $v_1(A_1) = (1+\eps)^{10}$ and $v_2(A_2) = 2+\eps$, and under $B$ are $v_1(B_1) = 2+\eps$ and $v_2(B_2) = (1+\eps)^{10}$.

Now consider a fractional allocation $\x$ given by the following assignment matrix:
\[
\x = \begin{bmatrix}
		\begin{array}{l|ccc}
		& g_1 & g_2 & g_3\\
	    \hline
		a_1 & 1 & 1 - (1+\eps)^{-9} & 0\\
		a_2 & 0 & (1+\eps)^{-9} & 1\\
		\end{array}
\end{bmatrix}
\]

The utilities of the agents under $\x$ are $v_1(\x_1) = (1+\eps)^{10}$ and $v_2(\x_2) = 2+2\eps$, implying that $\x$ Pareto dominates $A$. For a similar reason, the fractional allocation $\y$ given by
\[
\y = \begin{bmatrix}
		\begin{array}{l|ccc}
		& g_1 & g_2 & g_3\\
	    \hline
		a_1 & 0 & (1+\eps)^{-9} & 1\\
		a_2 & 1 & 1 - (1+\eps)^{-9} & 0\\
		\end{array}
\end{bmatrix}
\]
Pareto dominates the allocation $B$. Hence, the above instance has no \EQx{} and \fPO{} allocation. Note that the above instance does admit an \EQx{} and \PO{} allocation (\Cref{prop:Leximin_EQx+PO_Positive_valuations}) as both $A$ and $B$ satisfy these two properties.
\end{proof}
\begin{remark}
The instance in \Cref{prop:Nonexistence_EQx+fPO} also rules the existence of an \EFx{} and \fPO{} allocation.
\end{remark}

\subsection{Proof of Theorem~\ref{thm:EQ1+EF1+PO_Polytime_BinaryVals}}
\label{subsec:Proof_EQ1+EF1+PO_Polytime_BinaryVals}

The proof of \Cref{thm:EQ1+EF1+PO_Polytime_BinaryVals} relies on a series of intermediate results (\Cref{lem:Edge_Decomposition_Lemma,lem:NashOptimal_Binary_Utility_Profile,lem:NashOptimality_Of_EQ1+PO_BinaryVals,lem:NashOptimal_EQ1+PO_BinaryVals}) that are presented below. We will start with a simple observation concerning directed graphs. It will be convenient to define, for any given directed graph $G = (V,E)$, the sets $\E_{+} \coloneqq \{v \in V : \outdeg_G(v) > \indeg_G(v)\}$ and $\E_{-} \coloneqq \{v \in V : \indeg_G(v) > \outdeg_G(v)\}$, where $\indeg_G(v)$ and $\outdeg_G(v)$ denote the indegree and outdegree of a vertex $v$, respectively.

\begin{restatable}[\textbf{Edge-decomposition lemma}]{lemma}{EdgeDecompositionLemma}
 \label{lem:Edge_Decomposition_Lemma}
 Let $G = (V,E)$ be a directed graph. Then, the set of edges $E$ can be partitioned into a set of directed paths $\P = \{P_1,P_2,\dots\}$ and a set of directed cycles $\C = \{C_1,C_2,\dots\}$ such that each path in $\P$ starts at a vertex in $\E_{+}$ and ends at a vertex in $\E_{-}$.
\end{restatable}
\begin{proof}
Consider a directed graph $H = (V_H,E_H)$ derived from the graph $G$ as follows: For every vertex $i \in V$ of the graph $G$, create $\Delta \coloneqq \max\{\indeg_G(i), \outdeg_G(i)\}$ copies of $i$, say $i_1,i_2,\dots,i_\Delta$. The vertex set $V_H$ of the graph $H$ is the union of all such copies of all the vertices in $V$. The set of edges $E_H$ is defined as follows: For every vertex $i \in V$, fix an arbitrary one-to-one correspondence between the incoming (respectively, outgoing) edges of $i$ in $G$ and the set $i_1,i_2,\dots,i_{\indeg_G(i)}$ (respectively, the set $i_1,i_2,\dots,i_{\outdeg_G(i)}$). For every edge $(i,k) \in E$, we add an edge between the corresponding copies of $i$ and $k$. This completes the construction of the graph $H$.

Notice that each vertex in $H$ has at most one incoming edge and at most one outgoing edge. Therefore, the set of edges $E_H$ can be partitioned into a set of directed paths $\P = \{P_1,P_2,\dots\}$ and a set of directed cycles $\C = \{C_1,C_2,\dots\}$ that are vertex-disjoint (in $H$). Since the edges in $H$ are in a one-to-one correspondence with the edges in $G$, the sets $\P$ and $\C$ also constitute a partition of the edges in $G$ (not necessarily vertex-disjoint in $G$).

Finally, notice that each path $P \in \P$ in $H$ must start at a vertex $i \in V_H$ with $\outdeg_H(i) = 1$ and $\indeg_H(i) = 0$, and must end at a vertex $k \in V_H$ with $\outdeg_H(k) = 0$ and $\indeg_H(k) = 1$. By construction, a vertex in $H$ has outdegree $1$ and indegree $0$ if and only if it is a copy of a vertex in the set $\E_{+}$ in $G$. Similarly, a vertex in $H$ with outdegree $0$ and indegree $1$ is a copy of a vertex in the set $\E_{-}$. Therefore, any path in $\P$, when projected to the graph $G$, must start at a vertex in $\E_{+}$ and end at a vertex in $\E_{-}$, as desired.
\end{proof}

Given a \emph{starting} allocation $A$ and a \emph{target} allocation $B$, define a directed graph $G = (V,E)$ over the set of agents (i.e., $V = [n]$) as follows: For every good $g \in A_i \cap B_k$, add an edge directed from vertex $i$ to vertex $k$. We call $G$ the \emph{transformation graph} from $A$ to $B$. Notice that a directed path in $G$ denotes a chain of pairwise swaps of goods. Similarly, a cycle in $G$ denotes a cyclic exchange of goods. Also notice that $G$ has at most $m$ edges.\footnote{A similar construction was used in the analysis of an algorithm for computing Nash optimal allocations for binary valuations \citep{BKV18greedy}.}

From \Cref{lem:Edge_Decomposition_Lemma}, we know that $E$ can be partitioned into a set of directed paths $\P = \{P_1,P_2,\dots\}$ and a set of directed cycles $\C = \{C_1,C_2,\dots\}$ such that each path starts at a vertex in $\E_{+}$ and ends at a vertex in $\E_{-}$, where $\E_{+} \coloneqq \{i \in [n] : |A_i| > |B_i|\}$ and $\E_{-} \coloneqq \{i \in [n] : |B_i| > |A_i|\}$.

For binary valuations, an allocation is Pareto optimal if and only if each good is assigned to an agent that approves it. Therefore, if $A$ and $B$ are Pareto optimal and the valuations are binary, we also have that $\E_{+} = \{i \in [n] : v_i(A) > v_i(B)\}$ and $\E_{-} = \{i \in [n] : v_i(B) > v_i(A)\}$. We formalize this observation in \Cref{cor:Edge_Decomposition_Pareto_Optimal_Binary}.

\begin{restatable}{corollary}{EdgeDecompositionParetoOptimalBinary}
 \label{cor:Edge_Decomposition_Pareto_Optimal_Binary}
 Let $A$ and $B$ be two Pareto optimal allocations for a given fair division instance with binary valuations, and let $G = (V,E)$ be the transformation graph from $A$ to $B$. Then, the set of edges $E$ can be partitioned into a set of at most $m$ directed paths $\P = \{P_1,P_2,\dots\}$ and a set of directed cycles $\C = \{C_1,C_2,\dots\}$ such that each path starts at a vertex in $\E_{+}$ and ends at a vertex in $\E_{-}$, where $\E_{+} \coloneqq \{i \in [n] : v_i(A) > v_i(B)\}$ and $\E_{-} \coloneqq \{i \in [n] : v_i(B) > v_i(A)\}$. 
\end{restatable}

\Cref{cor:Edge_Decomposition_Pareto_Optimal_Binary} will be useful in proving \Cref{lem:NashOptimality_Of_EQ1+PO_BinaryVals,lem:NashOptimal_Binary_Utility_Profile}.

\begin{restatable}{lemma}{EQonePOIsNashOptimalBinaryVals}
 \label{lem:NashOptimality_Of_EQ1+PO_BinaryVals}
 Let $\I$ be a fair division instance with binary valuations. If there exists an allocation $A$ that is \EQ{1} and \PO{} for $\I$, then $A$ must also be Nash optimal.
\end{restatable}
\begin{proof}
Suppose, for contradiction, that $A$ is \EQ{1} and \PO{} but not Nash optimal, and let $B$ be some Nash optimal allocation. Let $G = (V,E)$ be the transformation graph from $A$ to $B$, and let $\E_{+} \coloneqq \{i \in [n] : v_i(A) > v_i(B)\}$ and $\E_{-} \coloneqq \{i \in [n] : v_i(B) > v_i(A)\}$. Since both $A$ and $B$ are \PO{}, we can apply \Cref{cor:Edge_Decomposition_Pareto_Optimal_Binary} to obtain a partition of the set of edges $E$ into a set of at most $m$ directed paths $\P = \{P_1,P_2,\dots\}$ and a set of directed cycles $\C = \{C_1,C_2,\dots\}$ such that each path starts at an agent in $\E_{+}$ and ends at an agent in $\E_{-}$. Notice that $\P \neq \emptyset$, since each cyclic exchange in $\C$ amounts to no change in Nash social welfare. 

Let $A^{1} \coloneqq A$ and $A^{i+1} \coloneqq P_i(A^{i})$, where $P_1,P_2,\dots$ are the paths in $\P$, and $P_i(A^{i})$ is the allocation obtained by performing the pairwise swaps along the path $P_i$ starting from $A^{i}$. Then, the quantity $\NSW(B) - \NSW(A)$ can be written as the following telescoping sum:
\begin{align*}
\NSW(B) - \NSW(A) = \sum_{P_i \in \P} \NSW(A^{i+1}) - \NSW(A^{i}).
\end{align*}

Since $A$ is not Nash optimal, we have that $\NSW(B) > \NSW(A)$. Therefore, there must exist some path $P_i \in \P$ such that $\NSW(A^{i+1}) > \NSW(A^{i})$, where $A^{i+1} = P_i(A^{i})$. Say $P_i$ starts at an agent $s \in \E_{+}$ and ends at an agent $t \in \E_{-}$. Notice that the utility of every agent other than $s$ and $t$ remains unchanged between $A^{i}$ and $A^{i+1}$. Therefore, the inequality $\NSW(A^{i+1}) > \NSW(A^{i})$ implies that $v_s(A^{i+1}_s) \cdot v_t(A^{i+1}_t) > v_s(A^{i}_s) \cdot v_t(A^{i}_t)$. Since the valuations are binary, we have $v_s(A^{i+1}_s) = v_s(A^{i}_s) - 1$ and $v_t(A^{i+1}_t) = v_t(A^{i}_t) + 1$. Substituting these relations in the above inequality gives $v_s(A^{i}_s) > v_t(A^{i}_t) + 1$. Notice that while performing pairwise swaps according to the paths in $\P$, the utility of the agents in $\E_{+}$ can never increase (and that of the agents in $\E_{-}$ can never decrease). Therefore, $v_s(A_s) \geq v_s(A^{i}_s)$ and $v_t(A_t) \leq v_t(A^{i}_t)$. We therefore get $v_s(A_s) > v_t(A_t) + 1$, which contradicts that $A$ is \EQ{1}. Hence, $A$ must be Nash optimal.
\end{proof}

Given any allocation $A$, its \emph{utility profile} is defined as the ordered $n$-tuple of agents' utilities $(v_1(A_1),v_2(A_2),\dots,v_n(A_n))$. Two allocations $A$ and $B$ are said to have \emph{equivalent utility profiles} if there is a permutation of the agents $\sigma: [n] \rightarrow [n]$ such that for every $i \in [n]$, we have $v_{\sigma(i)}(A_{\sigma(i)}) = v_i(B_i)$.

\begin{restatable}{lemma}{NashOptimalBinaryValsUtilityProfile}
 \label{lem:NashOptimal_Binary_Utility_Profile}
 Let $\I$ be a fair division instance with binary valuations and let $A$ and $B$ be two Nash optimal allocations. Then, $A$ and $B$ have equivalent utility profiles.
\end{restatable}
\begin{proof}
Let $G = (V,E)$ be the transformation graph from $A$ to $B$, and let $\E_{+} \coloneqq \{i \in [n] : v_i(A) > v_i(B)\}$ and $\E_{-} \coloneqq \{i \in [n] : v_i(B) > v_i(A)\}$. Since both $A$ and $B$ are \PO{}, we can apply \Cref{cor:Edge_Decomposition_Pareto_Optimal_Binary} to obtain a partition of the set of edges $E$ into a set of directed paths $\P = \{P_1,P_2,\dots\}$ and a set of directed cycles $\C = \{C_1,C_2,\dots\}$ such that each path starts at an agent in $\E_{+}$ and ends at an agent in $\E_{-}$. Notice that if $\P = \emptyset$, then the claim follows because Nash social welfare remains unchanged under any cyclic exchange. Therefore, for the rest of the proof, we will assume that $\P \neq \emptyset$.

We claim that for every path $P_i \in \P$, $\NSW(P_i(A)) = \NSW(A)$, where $P_i(A)$ is the allocation obtained by performing the pairwise swaps along the path $P_i$ starting from $A$. Indeed, since $A$ is Nash optimal, we know that $\NSW(P_i(A)) \leq \NSW(A)$. Suppose, for contradiction, that $\NSW(P(A)) < \NSW(A)$ for some path $P \in \P$. By reindexing, we can write $P_1 \coloneqq P$, and therefore $\NSW(P_1(A)) < \NSW(A)$. Analogously to the proof of \Cref{lem:NashOptimality_Of_EQ1+PO_BinaryVals}, define $A^{1} \coloneqq A$ and $A^{i+1} \coloneqq P_i(A^{i})$, where $P_1,P_2,\dots$ are the paths in $\P$, and $P_i(A^{i})$ is the allocation obtained by performing the pairwise swaps along the path $P_i$ starting from $A^{i}$. We can once again express $\NSW(B) - \NSW(A)$ as the following telescoping sum:
\begin{align*}
\NSW(B) - \NSW(A) & = \sum_{P_i \in \P} \NSW(A^{i+1}) - \NSW(A^{i})\\
& = \NSW(P_1(A)) - \NSW(A) + \sum_{P_i \in \P \setminus \{P_1\}} \NSW(A^{i+1}) - \NSW(A^{i}).
\end{align*}

Since $A$ and $B$ are both Nash optimal, $\NSW(B) = \NSW(A)$. Therefore, if $\NSW(P_1(A)) < \NSW(A)$, then there must exist some path $P_i \in \P \setminus \{P_1\}$ such that $\NSW(A^{i+1}) > \NSW(A^{i})$, where $A^{i+1} = P_i(A^{i})$. Proceeding as in the proof of \Cref{lem:NashOptimality_Of_EQ1+PO_BinaryVals}, we get that if $P_i$ starts at $s \in \E_{+}$ and ends at $t \in \E_{-}$, then $v_s(A_s) > v_t(A_t) + 1$.

Consider the allocation $A' \coloneqq P_i(A)$. Notice that $v_s(A'_s) = v_s(A_s) - 1$ and $v_t(A'_t) = v_t(A_t) + 1$. Then,
\begin{align*}
\NSW(A') > \NSW(A) & \Leftrightarrow v_s(A'_s) \cdot v_t(A'_t) > v_s(A_s) \cdot v_t(A_t) \\
& \Leftrightarrow (v_s(A_s)-1) \cdot (v_t(A_t)+1) > v_s(A_s) \cdot v_t(A_t)\\
& \Leftrightarrow v_s(A_s) > v_t(A_t) + 1.
\end{align*}
We know from preceding analysis that the last inequality is true. Therefore, we get that $\NSW(A') > \NSW(A)$, which contradicts that $A$ is Nash optimal. Thus, for every path $P_i \in \P$, $\NSW(P_i(A)) = \NSW(A)$.

The above observation implies that for any path $P \in \P$ that starts at $s \in \E_{+}$ and ends at $t \in \E_{-}$, we have $v_s(A_s) = v_t(A_t) + 1$. After the reallocation along $P$, we obtain another Nash optimal allocation $A' = A(P)$ such that $v_t(A'_t) = v_s(A'_s) + 1$. Notice that $A$ and $A'$ have equivalent utility profiles with respect to the permutation $\sigma$ such that $\sigma(s) = t$, $\sigma(t) = s$, and $\sigma(i)=i$ for all $i \in [n] \setminus \{s,t\}$. The lemma now follows from induction over the paths in $\P$.
\end{proof}

\begin{restatable}{lemma}{NashIsEQonePOBinaryVals}
 \label{lem:NashOptimal_EQ1+PO_BinaryVals}
 Let $\I$ be a fair division instance with binary valuations. If $\I$ admits some \EQ{1} and \PO{} allocation, then every Nash optimal allocation must satisfy \EQ{1}.
\end{restatable}
\begin{proof}
Suppose, for contradiction, that there exists a Nash optimal allocation $A$ that is not \EQ{1}, even though there exists another allocation $B$ that is \EQ{1} and \PO{}. Since $A$ violates \EQ{1}, there must exist a pair of agents $h,k \in [n]$ such that for every good $j \in A_k$, we have $v_h(A_h) < v_k(A_k \setminus \{j\})$. Since the valuations are binary and additive and $A$ is \PO{}, we have $v_{k,j} = 1$, and therefore $v_h(A_h) < v_k(A_k) - 1$.

From \Cref{lem:NashOptimality_Of_EQ1+PO_BinaryVals}, we know that $B$ is Nash optimal. Then, by \Cref{lem:NashOptimal_Binary_Utility_Profile}, $B$ and $A$ must have equivalent utility profiles, i.e., there exists a permutation $\sigma : [n] \rightarrow [n]$ such that for every $i \in [n]$, we have $v_{\sigma(i)}(B_{\sigma(i)}) = v_i(A_i)$. This gives $v_{\sigma(h)}(B_{\sigma(h)}) < v_{\sigma(k)}(B_{\sigma(k)}) - 1$, which contradicts that $B$ is \EQ{1}.
\end{proof}

We are now ready to prove \Cref{thm:EQ1+EF1+PO_Polytime_BinaryVals}.

\EQoneEFonePOPolytimeBinaryVals*
\begin{proof}
It is known that for binary valuations, a Nash optimal allocation can be computed in polynomial time \citep{DS15maximizing,BKV18greedy}. If this allocation satisfies \EQ{1}, then we have the desired allocation (since a Nash optimal allocation is also \EF{1}+\PO{}; \citealt{CKM+16unreasonable}). Otherwise, using the contrapositive of \Cref{lem:NashOptimal_EQ1+PO_BinaryVals}, we know that no allocation satisfies \EQ{1} and \PO{}.
\end{proof}

\subsection{Experimental Results for Envy-Freeness}
\label{subsec:Expt_Goods_EFx}

\begin{figure}
\centering
    \includegraphics[width=0.75\linewidth]{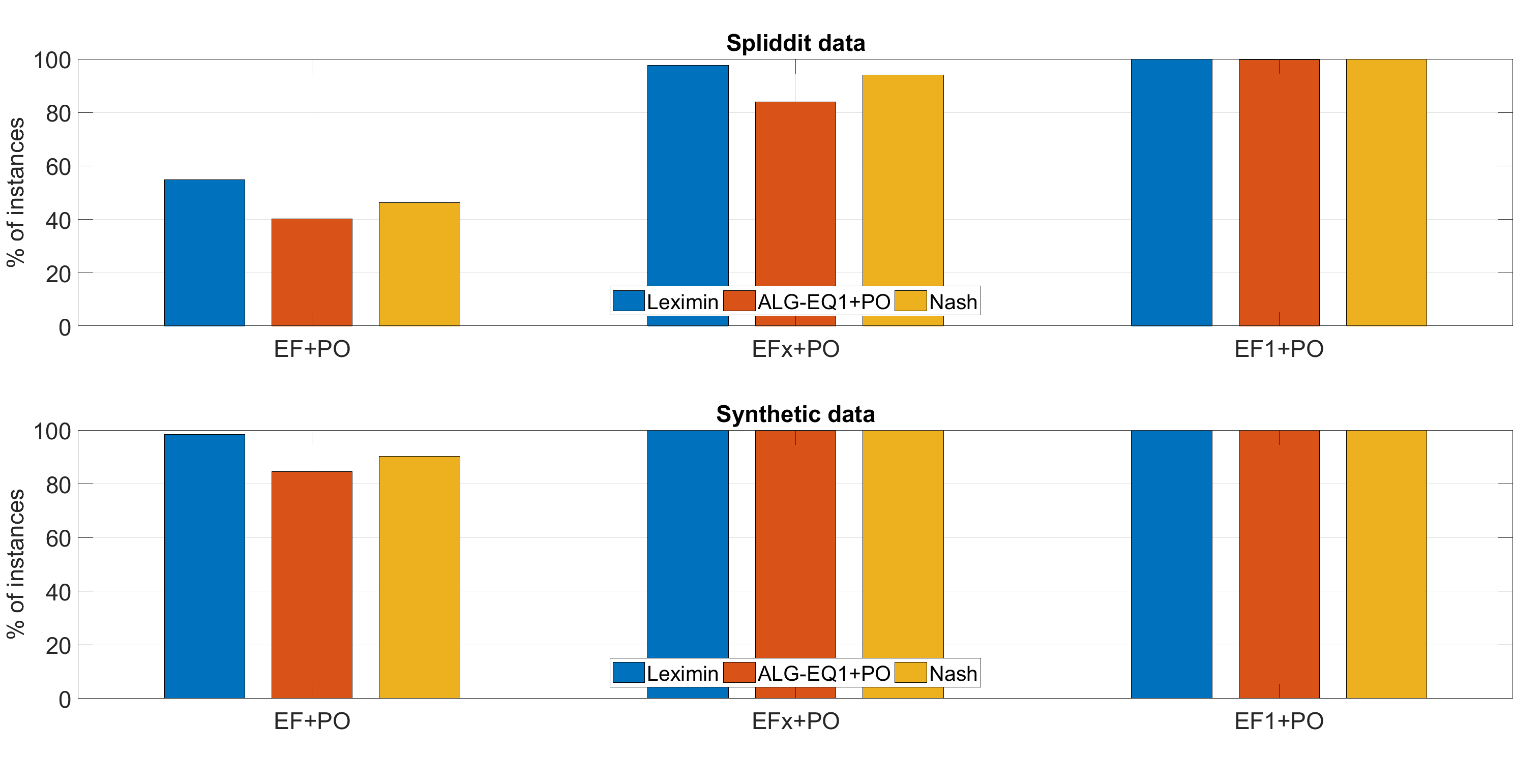}
\caption{Experimental results for envy-freeness and its relaxations.}
\label{fig:ef-experiments}
\end{figure}

In this section we discuss additional experimental results that pertain only to (relaxations of) envy-freeness, not equitability. Figure~\ref{fig:ef-experiments} presents histograms of the same form as Figure~\ref{fig:Experimental_Results_Goods_Main} for the following combinations of properties: \EF{}+\PO{}, \EF{1}+\PO{}, and \EF{x}+\PO{}.

We begin by noting that all algorithms satisfy \EF{1}+\PO{} on all synthetic instances and almost all Spliddit instances. Of the 914 Spliddit instances, \Leximin{} satisfies \EF{1}+\PO{} on 913 and \AlgEQonePO{} on 912 of them (MNW is theoretically guaranteed to always satisfy \EF{1}+\PO{}). 

Interestingly, despite being guaranteed to satisfy \EF{1}, MNW is actually outperformed by \Leximin{} with respect to both \EF{} and \EF{x}. \Leximin{} satisfies \EF{} on 55\% of Spliddit instances and 98\% of synthetic instances, whereas MNW satisfies \EF{} on only 46\% of Spliddit instances and 90\% of synthetic instances. \AlgEQonePO{} performs the worst from the perspective of envy, which is not surprising given that it is specifically designed with equitability (up to one good) in mind. All algorithms perform better with respect to the relaxed property \EF{x}, but the relative performance of the algorithms is preserved, with \Leximin{} satisfying \EF{x} on strictly more instances than MNW, which in turn performs better than \AlgEQonePO{}.

In contrast to the equitability-related properties investigated in Section~\ref{sec:Experiments}, all algorithms achieve envy-related properties more frequently on synthetic instances than on Spliddit instances. We speculate that this is due to greater correlation between agents' valuations for Spliddit data than in our Dirichlet model, making Spliddit instances `harder' from the perspective of removing envy.

\subsection{Miscellaneous Examples}
\label{subsec:Misc_Examples}

\begin{example}
\label{eg:Relax_And_Round_Fails}
This example presents an instance where every rounding of the fractional maximin allocation (i.e., a fractional allocation that maximizes the minimum utility) violates \EQ{1}. Consider the following instance with three agents $a_1,a_2,a_3$ and their (additive) valuations over three goods $g_1,g_2,g_3$ as below:

\begin{table}[ht]
\centering
\begin{tabular}{ c|ccc }
	& $g_1$ & $g_2$ & $g_3$\\ \hline
  $a_1$ & $\nicefrac{1}{4} - \nicefrac{\eps}{2}$ & $\nicefrac{1}{4} - \nicefrac{\eps}{2}$ & $\nicefrac{1}{2} + \eps$\\
  $a_2$ & $\eps$ & $\eps$ & $1 - 2\eps$ \\
  $a_3$ & $\eps$ & $\eps$ & $1 - 2\eps$ \\
\end{tabular}
\end{table}

The fractional maximin allocation assigns $g_1$ and $g_2$ to $a_1$, and splits $g_3$ equally between $a_2$ and $a_3$. In any rounding of this allocation, either $a_2$ or $a_3$ will get an empty bundle, thus creating an \EQ{1} violation with $a_1$.

We remark that the above example also shows that \emph{any} rounding of the fractional MNW allocation (i.e., a fractional allocation that maximizes the geometric mean of the utilities) violates \EF{1}. A similar implication was previously shown by \citet{CKM+16unreasonable} by means of a more complicated example involving $3$ agents and $31$ goods.
\end{example}

\begin{example}
This example shows that allocations other than a \Leximin{}-optimal can satisfy \EQ{1} and \PO{}. Suppose there are three goods $\{g_1,g_2,g_3\}$ and two agents $\{1,2\}$ with identical valuations given by $v_{g_1} = 2$ and $v_{g_2} = v_{g_3} = 1$. The allocation $A = (A_1,A_2)$ with $A_1 = \{g_1,g_2\}$ and $A_2 = \{g_3\}$ satisfies \EQ{1} and \PO{}, but is not \Leximin{}-optimal.
\end{example}

\begin{example}
This example shows that when we allow the valuations to be zero-valued, a \Leximin{}-optimal allocation need not be \EQ{1}, even when there exists another allocation that is \EQ{1} and \PO{}. Consider the following instance with three agents and eight goods:
\begin{table}[ht]
\centering
\begin{tabular}{ c|ccc }
	& $g_1$ & $g_2,g_3$ & $g_4,g_5,g_6,g_7,g_8$\\ \hline
  $a_1$ & $7$ & $(0,0)$ & $(0,0,0,0,0)$\\
  $a_2$ & $0$ & $(5,5)$ & $(2,2,2,2,2)$ \\
  $a_3$ & $0$ & $(5,5)$ & $(2,2,2,2,2)$ \\
\end{tabular}
\end{table}

The above instance has exactly two \Leximin{}-optimal allocations, $\{g_1\},\{g_2,g_3\},\{g_4,\dots,g_8\}$ and $\{g_1\},\{g_4,\dots,g_8\},\{g_2,g_3\}$, both of which violate $\EQ{1}$. Yet, there exists an \EQ{1} and \PO{} allocation $\{g_1\},\{g_2,g_4,g_5\},\{g_3,g_6,g_7,g_8\}$. Note that \PO{} follows from the fact that $g_1$ is uniquely valued by agent $1$, and that $g_2,\dots,g_8$ constitute a sub-instance with identical valuations for the agents $a_2$ and $a_3$.
\end{example}

\subsection{Approximating Max-Min Fairness}
\label{sec:maxmin}

We now ask whether our (approximate) notions of equitability, in conjunction with Pareto optimality, provide any approximation guarantee for the \SantaClaus{} problem~\citep{BD05allocating,BS06santa}. Intuitively, both (approximate) equitability and \SantaClaus{} have egalitarian goals: The former aims to minimize the disparity between the best-off and worst-off individuals, while the latter aims to maximize the utility of the worst-off agent. It is therefore pertinent to ask whether and how well do the solutions of one problem work for the other.

Formally, an instance of \SantaClaus{} consists of a fair division instance (with additive valuations) and some $r \in \mathbb{R}$, and asks whether it is possible to allocate the goods so that every agent has utility at least $r$. For a given instance, denote by \OPT{} the highest value of $r$ for which the \SantaClaus{} instance has a solution. For any $b \in [0,1]$, we say that an allocation $A$ is a $b$-approximation if $v_i(A_i) \ge b \cdot \OPT{}$ for all $i \in [n]$. It is known that computing a $b$-approximation for \SantaClaus{} for any fixed $b > \frac{1}{2}$ is strongly \NPhard{}~\citep{BD05allocating}.

We first observe that an (arbitrary) allocation satisfying \EQ{x} and \PO{} might fail to provide any non-trivial approximation to \SantaClaus{}.

\begin{example}
Consider an instance with two agents $a_1,a_2$ and two goods $g_1,g_2$, with $v_{1,1}=K$, $v_{1,2}=K-1$, $v_{2,1}=K-1$, and $v_{2,2}=1$, for some positive $K$. The allocation $A = (A_1,A_2)$ given by $A_1 = \{ g_1 \}$ and $A_2 = \{ g_2 \}$ is \EQx{} and \PO{}. Under $A$, agent $a_2$ has a utility of $1$, whereas it is possible for the agents to swap bundles and receive a utility of $\OPT{} = K-1$ each. The approximation ratio of $\frac{1}{K-1}$ can be made arbitrarily small for a suitably large $K$.
\label{ex:EQX+PO_no_mm_approximation}
\end{example}

We can, however, guarantee a non-trivial approximation by imposing additional structure on the \EQx{}+\PO{} solution. In particular, \Cref{thm:mm_approximation} achieves a better approximation factor that depends on the number of goods allocated to the agent with highest utility.

\begin{restatable}[\textbf{\SantaClaus{} Approximation for \EQ{x} and \PO{} Allocations}]{theorem}{mmapproximation}
	Any allocation that satisfies \EQx{}+\PO{} and allocates $c$ goods to an agent with highest utility provides a $(1-1/c)$-approximation to the \SantaClaus{} problem. Furthermore, this bound is tight.
	\label{thm:mm_approximation}
\end{restatable}
\begin{proof}
We will start by showing the approximation guarantee. Let $A$ be any \EQx{}+\PO{} allocation that assigns $c$ goods to an agent with the highest utility, say $i$ (thus, $i \in \arg\max_{k \in [n]} v_k(A_k)$ and $|A_i|=c$). Then, we must have that $v_i(A_i) \ge \OPT{}$, otherwise there exists an allocation $B$ for which $v_k(B_k) \ge \OPT{} > v_i(A_i) \ge v_k(A_k)$ for all $k \in [n]$, which violates \PO{}. 

Let $j \in A_i$ be the least positively valued good of agent $i$ in its bundle, i.e., $j \in \arg\min_{\ell \in A_i : v_{i,\ell} > 0} v_{i,\ell}$. Then, $v_{i,j} \le (1/c) \cdot v_i(A_i)$. The desired approximation now follows by invoking the \EQx{} condition: For any agent $k \in [n]$,
\begin{align*}
	v_k(A_k) \ge v_i(A_i)-v_{i,j} \ge (1-\frac{1}{c})v_i(A_i) \ge (1-\frac{1}{c}) \OPT{}.
\end{align*}

We will now show that the above bound is tight. Consider the instance in \Cref{ex:EQX+PO_no_mm_approximation}, with two additional copies of good $g_1$; thus four goods overall. With slight abuse of notation, we will write $\{g_1,g_2,g_1,g_1\}$ to denote the set of goods. Notice that the allocation $A = (A_1,A_2)$ given by $A_1 = \{g_1,g_1\}$ and $A_2 = \{g_1,g_2\}$ is \EQx{}, \PO{}, and assigns $c=2$ goods to the highest utility agent. From the above analysis, $A$ achieves (at least) $\frac{1}{2}$-approximation to \SantaClaus{}. Furthermore, under $A$, agent $a_2$ receives a utility of $K$. However, the allocation $B = (B_1,B_2)$ with $B_1 = \{g_1,g_2\}$ and $B_2 = \{g_1,g_1\}$ shows that $\OPT \geq 2K-2$. For a suitably large $K$, the approximation ratio of $\frac{K}{2K-2}$ achieved by $A$ can be made arbitrarily close to $\frac{1}{2}$.
\end{proof}

\begin{remark}
	As an application of Theorem~\ref{thm:mm_approximation}, suppose that all valuations $v_{i,j}$ lie in the range $[1,c]$, and $m > c^2n$. The average utility is at least $c^2$, which means that the agent with the highest utility must receive at least $c$ goods (since each good is valued at no more than $c$). In this setting, any allocation that satisfies \EQ{x} and \PO{} therefore achieves a $(1-\frac{1}{c})$-approximation to \OPT{}.
\end{remark}

For identical valuations, we can give a tight bound on the approximation guarantee of an \EQ{1} allocation.

\begin{restatable}[\textbf{\SantaClaus{} Approximation for \EQ{1} Allocations and Identical Valuations}]{theorem}{EQonePOidenticalmmapproximation}
	For identical valuations, any allocation satisfying \EQ{1} provides a $\frac{1}{n}$-approximation to the \SantaClaus{} problem. Furthermore, this bound is tight.
	\label{thm:EQone+PO_identical_mm_approximation}
\end{restatable}

\begin{proof}
	Denote the common valuation function by $v$.
	We will first prove that no \EQ{1} allocation can do better than $\frac{1}{n}$. Consider an instance with $n$ agents and $m=2n-1$ goods. For every good $j \in \{ 1, \ldots, n \}$, $v_{j}=1$, and for every good $j \in \{ n+1, \ldots, 2n-1 \}$, $v_{j}=n$. Define an allocation $A$ with $A_i= \{ i, n+i \}$ for all $i \in \{1, \ldots, n-1\}$, and $A_n = \{ n \}$. Observe that $A$ is \EQ{1}. Also, agent $n$ has a utility of only 1, even though it is possible to give every agent a utility of $n$. 
	
	We next prove the approximation guarantee. Suppose that allocation $A$ satisfies \EQ{1}, and let $i$ be an agent with lowest utility under $A$. Then, for any $k \neq i$ such that $A_k \neq \emptyset$, there exists a good $g_k \in A_k$ such that 
	\begin{equation}
		\label{eq:EQ1}
		v(A_i) \ge v(A_k)-v_{g_k}. 
		\end{equation}
	
	Consider the alternative allocation $B$ with $B_k = \{ g_k \}$ for all $k \neq i$, and $B_i = [m] \setminus \cup_{k \neq i} \{g_k \}$. Note that $v(B_i)$ cannot be less than \OPT{}: If it were, then there would exist another allocation $C$ with $v(C_k) > v(B_i)$ for all agents $k$. But this is impossible, because there are at most $n-1$ goods with $v_j > v(B_i)$---namely, each of the $g_k$ goods---and the value of all other goods only sums to $v(B_i)$.
	
	Combining these facts, we get that
	\begin{equation*}
		v(B_i) = v(A_i) + \sum_{k \neq i} (v(A_k)-v_{g_k}) \le n \cdot v(A_i),
	\end{equation*}
	where the inequality follows from \Cref{eq:EQ1}. Rearranging gives us $v(A_i) \ge \frac{1}{n} v(B_i) \ge \frac{1}{n} \OPT{}$, as desired.
\end{proof}

\end{document}